\PassOptionsToPackage{pagebackref=true, colorlinks=true, citecolor=green!60!black, linkcolor=blue}{hyperref}
\documentclass[11pt]{article}

\usepackage{booktabs} 
\usepackage{algorithm}
\usepackage{algorithmic}
\usepackage{fullpage}
\usepackage{pgfplots}
\usepackage{pgfplotstable}
\usetikzlibrary{patterns}
\usepackage[utf8]{inputenc}
\usepackage{parskip}
\usepackage{xcolor}
\usepackage{url}
\usepackage{amsmath,amsthm,amsfonts,amssymb,microtype,bbm}
\usepackage{graphicx}
\usepackage{dsfont}
\usepackage{wrapfig}
\usepackage[numbers]{natbib}

\newtheorem{theorem}{Theorem}[section]
\newtheorem{lemma}[theorem]{Lemma}
\newtheorem{proposition}[theorem]{Proposition}
\newtheorem{definition}[theorem]{Definition}
\newtheorem{example}[theorem]{Example} 

\newtheorem{claim}[theorem]{Claim}
\newtheorem{assumption}[theorem]{Assumption}
\newtheorem{corollary}[theorem]{Corollary}
\newtheorem{observation}[theorem]{Observation}
\newtheorem{remark}[theorem]{Remark}

\newcommand{\xhdr}[1]{\vspace{0mm} \noindent{\bf #1}}
\newcommand{\vx}{\mathbf{x}}

\newcommand{\bbR}{\mathbb{R}}
\newcommand{\vw}{\mathbf{w}}
\newcommand{\calX}{\mathcal{X}}

\newcommand{\calG}{\mathcal{G}}

\newcommand{\eps}{\varepsilon}

\DeclareMathOperator{\E}{\mathbb{E}}

\newcommand{\hy}{\hat{y}}
\newcommand{\calD}{\mathcal{D}}
\newcommand{\calS}{\mathcal{S}}
\newcommand{\hvx}{\hat{\vx}}

\newcommand{\calI}{\mathcal{I}}
\newcommand{\vwst}{\vw^{\star}}

\newcommand{\tvw}{\widetilde{\vw}}
\newcommand{\SW}{\texttt{\upshape{SW}}}
\newcommand{\COST}{\texttt{\upshape{Cost}}}
\newcommand{\diff}{\texttt{\upshape{diff}}}
\newcommand{\VALUE}{\texttt{\upshape{Score}}}

\newcommand{\proxy}{\texttt{\upshape est}}
\newcommand{\ucalI}{\texttt{\upshape u}\calI}

\newcommand{\bA}{\bar{A}}
\newcommand{\PD}{\texttt{\upshape PD}}
\newcommand{\PSD}{\texttt{\upshape PSD}}

\newcommand\numberthis{\addtocounter{equation}{1}\tag{\theequation}}

\newcommand{\squishlist}{
   \begin{list}{$\bullet$}
    { \setlength{\itemsep}{0pt}      \setlength{\parsep}{3pt}
      \setlength{\topsep}{3pt}       \setlength{\partopsep}{0pt}
      \setlength{\leftmargin}{1.5em} \setlength{\labelwidth}{1em}
      \setlength{\labelsep}{0.5em} } }
\newcommand{\squishend}{  \end{list}  }

\usepackage{nicefrac}

\title{Information Discrepancy in Strategic Learning}

\author{
    Yahav Bechavod\thanks{School of Computer Science and Engineering, The Hebrew University. Email: \texttt{yahav.bechavod@cs.huji.ac.il}.}
    \and
    Chara Podimata\thanks{School of Engineering and Applied Sciences, Harvard University. Email: \texttt{podimata@g.harvard.edu}.}
    \and
    Zhiwei Steven Wu\thanks{School of Computer Science, Carnegie Mellon University. Email: \texttt{zstevenwu@cmu.edu}.}
    \and
    Juba Ziani\thanks{School of Industrial and Systems Engineering, Georgia Institute of Technology. Email: \texttt{juba.ziani@isye.gatech.edu}.}
}

\begin{document}

\maketitle

\begin{abstract}
We initiate the study of the effects of non-transparency in decision rules on individuals' ability to improve in strategic learning settings. Inspired by real-life settings, such as loan approvals and college admissions, we remove the assumption typically made in the strategic learning literature, that the decision rule is fully known to individuals, and focus instead on settings where it is inaccessible. In their lack of knowledge, individuals try to infer this rule by learning from their peers (e.g., friends and acquaintances who previously applied for a loan), naturally forming groups in the population, each with possibly different type and level of information regarding the decision rule. We show that, in equilibrium, the principal's decision rule optimizing welfare across sub-populations may cause a strong negative externality: the true quality of some of the groups can actually \emph{deteriorate}. On the positive side, we show that, in many natural cases, optimal improvement can be guaranteed simultaneously for all sub-populations. We further introduce a measure we term \emph{information overlap proxy}, and demonstrate its usefulness in characterizing the disparity in improvements across sub-populations. Finally, we identify a natural condition under which improvement can be guaranteed for all sub-populations while maintaining high predictive accuracy. We complement our theoretical analysis with experiments on real-world datasets.

\end{abstract}

\newpage
\section{Introduction}

Machine learning algorithms are increasingly used to make consequential decisions across a variety of domains, including loan approvals, {college admissions}, probation qualifications, {and hiring}. Given the high stakes of these decisions, individuals are incentivized to invest effort in changing their attributes, to obtain more favorable decisions. The evidence for such strategic adaptation from multiple domains (e.g., \citet{bjorkegren2020, newyorktest, reportcards, Pollution, slipperyfish}) has inspired a growing literature on \emph{strategic learning} that studies the interaction between learning algorithms and \emph{strategic} individuals (``agents''). 

Models in the strategic learning literature, however, typically make a \emph{full transparency} assumption---that is, the agents \emph{fully observe} the deployed 
{scoring or} decision rule \citep{hardtetal, dongetal, grinder, BechavodCausal, Shavitetal, HIV19, ManishKleinbergEffort}. In the context of credit scores, for example, this translates to the assumption that individuals know the deployed credit scoring rule in full detail. However, in reality, such a full-transparency assumption is often \emph{far-fetched}, and as many credit scoring rules are  proprietary, banks or financial agencies rarely make their machine learning model fully transparent to outsiders. Instead, they may only provide some \emph{labeled examples} (e.g., past applicants who were granted loans) or \emph{explanations} (e.g., ways to improve one's credit score).

{As the actual scoring rule in effect is not directly observable, agents naturally attempt to infer it using other sources of information, which may differ greatly across different individuals.} This is the case when the population is naturally clustered (due to e.g., their demographic, geographic, and cultural differences) and people have the tendency of \emph{observational learning} \citep{obs, imitation,Pathological,Conformity} ---that is, agents learn by observing others within their communities. For example, when applying for a loan at a specific bank, individuals may learn from the past experiences of their peers/friends (i.e., their applications and loan decisions) to gauge the decision rule. Hence, individuals from different peer-networks may form different ideas about the decision rule, which in turn can lead to disparities in strategic investments and outcomes. To make things worse, there is often a regulatory requirement that the \emph{same} decision rule be used on all sub-populations (due to e.g., the risk of redlining \cite{redline,rothstein2017color}), prohibiting the use of group-specific decision rules by the decision-maker for to mitigate the adverse effects of information discrepancy.

\subsection{Our Work}

Our work introduces the first framework to study the disparate effects of non-transparency in strategic learning on individuals' ability to improve. Below, we provide an overview of our contributions, and the roadmap of the paper.

\xhdr{Equilibrium Model.} We propose a model for the principal-agent interaction, when individuals from different sub-populations \emph{learn from their peers} (Sec. \ref{sec:model}). We then show how individuals from different sub-populations use the information available among their peers to form estimates of the decision rule, and compute the closed-form solutions for their and the principal's responses in equilibrium (Sec. \ref{sec:equilibrium-compute}).

\xhdr{Improvement Across Sub-Populations.}
Using our proposed model, we first prove a strong negative externality result: even if the principal deploys a decision rule that is optimized for maximizing the improvement across sub-populations, and individuals best-respond to the principal's rule, some sub-populations may still suffer \emph{deterioration} in their ``quality'' (true label). On the positive side, we show that improvement can be guaranteed simultaneously for all sub-populations under moderate conditions, e.g., when they have similar costs for effort exertion or when the extent to which they share the same information is minimal (Sec.~\ref{subsec:do-no-harm}). 

We then examine the extent to which information discrepancy regarding the deployed decision rule may result in disparity of improvement between sub-populations. We introduce the \emph{information overlap proxy} measure and prove it upper bounds this disparity. We conclude by characterizing the exact conditions for the disparity to vanish (Sec.~\ref{subsec:total-outcome-improvement}). 

Subsequently, we study how efficiently each of the sub-populations exerts their efforts in improving their quality.
For that, we introduce the \emph{per-unit improvement} (which measures the efficiency of the sub-populations effort exertions), and we identify moderate conditions so that individuals from all sub-populations exert their effort optimally (Sec.~\ref{subsec:per-unit-outcome-improvement}).

Finally, we consider a case where the principal interpolates between the objectives of outcome improvement and predictive accuracy. We identify a natural condition under which improvement in all sub-populations is guaranteed while maintaining high predictive accuracy (Sec.~\ref{subsec:joint} and App.~\ref{app:commuting}).

We further show that similar conclusions can be drawn in the general settings where the \emph{principal is a learner} who does not originally know the properties of the individuals' sub-populations, but has to \emph{learn} them instead (App. \ref{app:learning}).

\xhdr{Empirical Evaluation.} 
Our experiments on real-world datasets (\textsc{Taiwan-Credit} and \textsc{Adult}) validate our theoretical results, further highlighting the pivotal role access to information plays in strategic settings (Sec. \ref{sec:experiments}).

\subsection{Related Work}\label{subsec:related}

Our work is primarily related to three strands of literature on strategic behavior in learning (of independent interest is \citet{AlternativeMicro21}, who re-examine the standard assumptions around strategic learning).
The first one advocates that changes in the agents' original features are considered ``gaming'', hence the learner wishes to construct algorithms that are robust to such behavior~\citep{SPclassification1, SPclassification2, SPclassification3,PP04,Dekeletal,hamsandwich,RS17, ic-learning, Dalvietal, Bruckneretal,hardtetal,dongetal,grinder,StrategicPerceptron,PACstrategicclassification,ioannidisloiseau,horeletal,caietal}. When the learner optimizes for robustness to strategic behavior, the deployed algorithm has disparate impact on different sub-populations~\citep{HIV19,Milli19SocialCost}. \citet{randomness-noise-forc20} study the impact of randomized or noisy classifiers on mitigating inequalities, but focus on a single-dimensional case. In our work, we do not consider the disparate impact of ``robustness'', but rather of information disparities across groups. In a concurrent and independent work, \citet{StratClassDark} study non-transparency in strategic classification, characterize the ``price of opacity'', and show conditions under which fully transparent classifiers are the recommended policy. Our work is orthogonal as it studies the \emph{societal implications} of ``opacity''.

The second strand of literature advocates that machine learning algorithms should incentivize ``good'' strategic behavior (\emph{aka} \emph{improvements}) \citep{ManishKleinbergEffort,ustunrecourse19,Tabibianetal19,TsirtsisCounterfactual,liuetal,Haghtalabetal,Alonetal,strategicrecourse20,EqualizingRecourse}. Our work is most closely related to \citep{liuetal,Haghtalabetal,EqualizingRecourse}. \citet{liuetal} study the long-term impact of strategic learning on different sub-populations, but focus on decision rules that are fully known to the agents. \citet{Haghtalabetal} study social welfare maximization when the learner does not have full knowledge of the feature space of the agents, contrary to our model where the information discrepancies appear on the agents' side. \citet{EqualizingRecourse} minimize the difference in recourse across sub-populations, whereas we focus on a principal optimizing the social welfare.

The third strand concerns causality in strategic learning (\citet{MillerCausalInDisguise20,BechavodCausal,Shavitetal} and more broadly~\citet{PerformativePrediction}), where the learner tries to learn the causal relationship between agents' features and their labels/scores by leveraging the agents' strategic behavior. Importantly, in our setting, even if the principal knows the causal relationship perfectly, the disparate impact from the algorithm may still be unavoidable. Recently, \citet{AlternativeMicro21} re-examine the standard assumptions around strategic learning and they propose alternative formulations.

Our work is connected to a literature on social welfare and fairness \citep{HeidariWelfare,HeidariEffortUnfairness,WelfareandFairClassification}. \citet{HeidariWelfare} propose incorporating social welfare considerations to the standard loss minimization goal of machine learning; our focus differs due to the presence of information discrepancies. \citet{WelfareandFairClassification} study the social welfare implications that result from a fair classification algorithm and show that applying more strict fairness criteria that are codified as parity constraints, can worsen welfare outcomes across sub-populations; our point of view is reversed and looks at the effect of welfare maximization on fairness. \citet{HeidariEffortUnfairness} also study how agents in different sub-populations invest their efforts through observational learning, by imitating a model behavior within their group. We consider a different type of observational learning where agents try to infer the deployed rule instead. Further, while they focus on disparities in level of effort across groups, we focus on disparities in improvement.

More broadly, the fact that peer-influenced behavior might induce disparities in the absence of perfect information has been studied in economics and sociology (e.g., \citet{coatelowery,okafor,birds,bridging,exacerbate,SocialNetworks}). In this paper, however, we aim to formally understand this phenomenon \emph{in the context of strategic learning}. We also go beyond characterizing such disparities, to consider objectives such as efficient effort exertion and improvement while maintaining high accuracy.
\section{Model and Preliminaries}\label{sec:model}

We study a {Stackelberg} game between a \emph{principal} and a population of \emph{agents} comprised of $m$ sub-populations (``groups'') with different distributions over the feature space $\calX \subseteq \bbR^d$. We focus on the case $m=2$ for clarity, but our results extend to arbitrary $m$, as outlined in Appendix~\ref{app:gen-multiple-groups}. Let the groups be $G_1$ and $G_2$, with associated distributions {of feature vectors} $\calD_1$ and $\calD_2$ respectively over $\calX$. Let $\calS_1, \calS_2$ be the subspaces defined by the supports of $\calD_1, \calD_2$ respectively. Let $\Pi_1, \Pi_2\in \bbR^{d\times d}$ be the orthogonal projection matrices onto subspaces $\calS_1, \calS_2$ respectively. Let $\vwst \in \bbR^{d}$ denote the \emph{ground truth} linear {assessment} rule (which is \emph{known}\footnote{We relax this assumption in App.~\ref{app:learning}.} to the principal through past observations): i.e., for a feature vector $\vx$, the corresponding agent's expected \emph{true} ``quality'' is given by $\mathbb{E}[y \mid \vx] = \langle \vwst,\vx\rangle$. Note that, while $\vwst$ is optimal for prediction accuracy, it may not be the one maximizing the welfare across groups. This is because it is often worth incentivizing modifications of features that are easy for agents to improve, and features who can be modified by and benefit several groups.

The principal deploys a linear scoring rule $\vw \in \bbR^d$. Agent $i$ from group $g$ draws private feature vector $\vx_{g,i} \sim \calD_g$. Initially, agents from both groups have no information regarding $\vw$, so they simply report $\vx_{g,i}$ to the principal and receive scores $\hy_{g,i} = \langle \vw, \vx_{g,i} \rangle$. After enough agents from both groups have received scores for their reported features, the remaining agents use this past information (i.e., feature-predicted score tuples) to appropriately alter their feature vectors from $\vx \sim \calD_g$ to $\hvx(\vx;g)$.
Knowing that the ground truth assessment rule together with the scoring rule that the principal deploys are \emph{linear}, and given the fact that they are risk-averse, agents perform empirical risk minimization (ERM) on the peer-dataset comprised of the first unmodified $N_g \in \bbR_+$ samples $S_g = \{(\vx_{g,i}, \hy_{g,i})\}_{i \in [N_g]}$ to compute an estimate $\vw_{\proxy}(g)$ of the deployed scoring rule $\vw$. Running ERM is a natural choice given that the agents are risk-averse, fully rational, and have no other information.

Given original features $\vx$ and estimation rule $\vw_{\proxy}(g)$, each (myopically rational) agent chooses $\hvx(\vx;g)$ as the $\vx'$ that optimizes their underlying utility function (which a generalization of the standard utility function used in the literature on strategic classification) defined as
\begin{equation}\label{eq:utility-func-formal}
u\left(\vx, \vx' ;g \right) = \VALUE\left(\vx';g\right) - \COST\left(\vx, \vx';g \right) 
\end{equation}
where $\VALUE \left(\vx' ;g \right) = \langle \vw_{\proxy}(g), \vx' \rangle$ is the estimate\footnote{The actual value that the agent derives by reporting $\vx'$ is the outcome $\langle \vwst, \vx' \rangle$. But $\vwst$ is never revealed to the agent; the only information that she has is the estimate for the principal's $\vw$.} value the agent derives for reporting feature vector $\vx'$ and $\COST(\vx, \vx';g) = \frac{1}{2} (\vx' - \vx)^\top A_g (\vx' - \vx)$ is the agent's cost for modifying vector $\vx$ into $\vx'$. We call $A_g \in \bbR^{d \times d}$ the \emph{cost matrix} for group $g$, and assume it is \emph{positive definite} (\PD).\footnote{In turn, one can write $\COST(\vx, \vx';g) = \frac{1}{2} \Vert \sqrt{A_g} (\vx'-\vx)\Vert_2^2$ noting that $\sqrt{A_g}$ is well-defined.} Due to not restricting $A_g$ further, this cost function family is rather large and encapsulates some cost functions used in the literature on strategic classification (e.g., \citep{dongetal,StrategicPerceptron}. This functional form is a simple way to model important practical situations in which features can be modified in a correlated manner, and investing in one feature may lead to changes in other features.

At a high level, the utility in Eq.~\eqref{eq:utility-func-formal} captures the ``net gains'' that an agent obtains from spending effort to report $\vx'$, rather than $\vx$. Since $\hvx(\vx;g)$ is the best response coming from Eq.~\eqref{eq:utility-func-formal}, then $\hvx(\vx;g) = \arg \max_{\vx' \in \calX} u\left(\vx, \vx'; g \right)$. As we show in Section~\ref{sec:equilibrium-compute} the best response $\hvx(\vx;g)$ takes the form $\vx + \Delta_g(\vw)$ for a ``movement'' function\footnote{Slightly abusing notation $\vw$ is an argument of $\Delta_g(\cdot)$, but this is only used in the analysis.
The agents do not directly see $\vw$.} $\Delta_g(\vw)$ to be specified shortly. Putting everything together, the protocol in Algorithm~\ref{protocol} summarizes the principal-agent interaction. %

\begin{algorithm}[htbp]
\caption{Principal-Agent Interaction Protocol}
\label{protocol}
\begin{algorithmic}[1]
   \STATE Nature selects ground truth scoring function $\vwst$.
   \STATE \label{step:2}Learner deploys scoring rule $\vw \in \bbR^d$ (solution to Eq.~\eqref{eq:max-vw}), but does \emph{not} directly reveal it to the agents.
   \STATE \label{step:3}Agents from groups $g \in \{1, 2\}$ draw their (private) feature vectors $\vx \sim \calD_g$.
   \STATE \label{step:4}Given peer-dataset $S_g$, (private) feature vector $\vx$, utility function $u(\vx,\vx';g)$, agents \emph{best-respond} with feature vector $\hvx(\vx;g) = \arg \max_{\vx' \in \calX} u(\vx,\vx';g)$.
\end{algorithmic}
\end{algorithm}%

When it comes to the principal's behavior, we posit that the principal's objective is to maximize the agents' {average social welfare across groups (``social welfare'' for short), defined as the sum over groups of the average (over agents) and expected (over the randomness of the labels) improvement of their true (as measured by $\vwst$) labels}, after best-responding. In other words, the principal deploys the equilibrium scoring rule $\vw_{\SW}$:
\begin{equation}\label{eq:max-vw}
\begin{aligned}
    \vw_{\SW} &= \arg \max_{\vw': \|\vw'\|_2 \leq 1} \SW\left(\vw' \right)\\
    &= \arg \max_{\vw': \|\vw'\|_2 \leq 1} \sum_{g \in \{1,2\}} \E_{\vx \sim \calD_g} \left[\left \langle \hvx(\vx;g), \vwst \right\rangle \right]
\end{aligned}
\end{equation}
In Sec.~\ref{sec:characterizations}, we additionally consider a principal who wishes to trade-off predictive accuracy and social welfare.

We aim to study the improvement among groups, in the presence of information discrepancy, \emph{at a Stackelberg equilibrium} of our game. In other words, the principal and agents best respond to each other, with the principal acting first and committing to a rule in anticipation of the strategic best responses of agents. We quantify improvement using two notions: \emph{total} improvement and \emph{per-unit} improvement.
\begin{definition}\label{def:per-unit-outcome-improve}
For rule $\tvw \in \bbR^d$, we define the \emph{total improvement} (``improvement'') for group $g$ as:
\begin{align*}
\calI_g \left(\tvw \right) &= \left \langle \hvx(\vx;g), \vwst\right \rangle - \left\langle \vx, \vwst \right\rangle\\
&= \left \langle \vx + \Delta_g \left(\tvw \right), \vwst\right \rangle - \left\langle \vx, \vwst \right\rangle\\
&= \left\langle \Delta_g\left( \tvw\right),\vwst \right\rangle.
\end{align*}
For the same rule, the \emph{per-unit improvement} for group $g$ is:
\[
\ucalI_g \left( \tvw \right) = \calI_g \left(\frac{\Pi_g\tvw}{\left \Vert\Pi_g\tvw\right\Vert_2}\right) = \left\langle \Delta_g\left( \frac{\Pi_g \tvw}{\left \| \Pi_g \tvw\right\|_2} \right),\vwst \right\rangle.
\]
\end{definition}
The usefulness of defining the total improvement as one of our measures is clear. The per-unit improvement only considers the part of the deployed scoring rule that belongs in the relevant subspace of each group, and measures how efficient the direction of this rule projected onto the relevant subspace is at inducing improvement for the group. 

We focus on three objectives for the two groups: \emph{do-no-harm}, \emph{equality}, and \emph{optimality}.

\begin{definition}[Do-No-Harm]
A rule $\tvw$ causes \emph{no harm} for group $g$ if $\calI_g \left (\tvw \right) \geq 0$.
\end{definition}

\begin{definition}[Equality]
A rule $\tvw$ enforces \emph{group-equality} if: $\calI_1(\tvw) = \calI_2(\tvw)$.
\end{definition}

\begin{definition}[Optimality] \label{def:optimality}
A rule $\vw'$ enforces $g$'s \emph{group-optimality} if: $\vw' = \arg\max_{\tvw} \ucalI_g(\tvw)$.
\end{definition}

\begin{remark} \label{rem:per-unit}
We note that achieving optimality in per-unit improvement (Def. \ref{def:optimality}) is equivalent to guaranteeing, for a rule $\vw$, that no other $\vw'$ for which $\Vert\Pi_g\vw'\Vert_2 \leq \Vert\Pi_g\vw\Vert_2$, can induce greater improvement than $\vw$ does in group $g$.
\end{remark}

Based on these objectives, we quantify how much the equilibrium play in this strategic interaction exacerbates inequalities between the groups due to their information discrepancies, even in the best-case scenario, where the principal is optimizing the population's average welfare across groups.

\section{Equilibrium Computation}\label{sec:equilibrium-compute}

In this section, we compute the equilibrium plays. We first compute the agents' estimate rules ($\vw_{\proxy}(g)$), given the information from their own group. We then derive the closed form of their best-response $\hvx(\vx;g)$. Using these, we solve the principal's optimization problem of Eq.~\eqref{eq:max-vw}. App.~\ref{app:equilibrium-compute} contains the proofs of the section.

\subsection{Computing an Estimate of Principal's Scoring Rule} \label{subsec:surr-vw-g}

Recall that agents from each group $g$ run ERM on their peer dataset $S_g$ to derive their \emph{estimated} decision rule $\vw_{\proxy}$. We posit that the agents are \emph{risk averse} ---that is, they prefer ``certain'' outcomes, rather than betting on uncertain ones. In our setting, this corresponds to agents taking the \emph{minimum norm} ERM to break ties, since they only wish to move in directions that can surely improve their outcome. Note that if agents invest efforts outside of the informational subspace, this could result in them not improving their outcome further (or even worse - deteriorating their outcome), while still incurring a cost. Formally, the agents compute $\vw_{\proxy}(g)$ as:%
\begin{align*}
    &\vw_\proxy(g) = \arg \min_{\tvw \in W} \|\tvw\|_2^2, \numberthis{\label{eq:erm}}
    \\ 
    &\text{s.t.} \; W = \left\{ w: w  = \arg\min_{\vw'} \sum_{i \in [N_g]} \left(\vx_{g,i}^\top \vw' - \hy_{g,i} \right)^2\right\} 
\end{align*}

When agents use ERM, we can state their estimate rule in closed form. 

\begin{lemma}\label{lem:inferred-rule}
Agents from group $g$ using ERM compute the estimate rule $\vw_\proxy(g) = \Pi_g \vw$.
\end{lemma}

\subsection{Closed Form of Agents' Best-Response} \label{subsec:closedform}

Slightly abusing notation, the agents' value becomes: $\VALUE(\vx,\vx'; g ) =  \langle \vw_{\proxy}(g), \vx'  \rangle$, which equals $\langle \Pi_g \vw, \vx'  \rangle$ from Lemma~\ref{lem:inferred-rule}. So, the agents' utility (Eq.~\eqref{eq:utility-func-formal}) becomes: 
\begin{equation}\label{eq:utility-simplified}
    u\left( \vx, \vx' ; g\right) = \left\langle \Pi_g \vw, \vx' \right \rangle - \frac{1}{2} \left\|\sqrt{A_g} \left(\vx' - \vx \right) \right\|^2
\end{equation}

\begin{lemma}\label{lem:br-computation}
The best-response of an agent from group $g$ with feature vector $\vx$ is: $\hvx(\vx;g) = \vx + A_g^{-1} \Pi_g \vw$. We write $\Delta_g(\vw) 
\triangleq \vx + \Delta_g(\vw)$.
\end{lemma}

In equilibrium, the principal \emph{knows} that the agent's best-response {as a function} of their private feature vector $\vx$ is given by Lemma~\ref{lem:br-computation}. We use this next when solving the principal's optimization problem.

\subsection{Principal's Chosen Scoring Rule in Equilibrium}\label{subsec:principal-prob}

Using the fact that the principal can compute $\Delta_g(\vw)$ for any group $g$, we can obtain a closed form solution for the principal's chosen rule $\vw$ (i.e., the solution to Eq.~\eqref{eq:max-vw}).

\begin{lemma}\label{lem:sw-opt-quasi}
The principal's scoring rule that maximizes the social welfare in equilibrium is:
\begin{equation}\label{eq:principal-vw2}
   \vw_{\SW} = \frac{ \left(\Pi_1 A_1^{-1} + \Pi_2 A_2^{-1} \right) \vwst}{\left\|\left(\Pi_1 A_1^{-1}+ \Pi_2 A_2^{-1} \right) \vwst \right\|}
\end{equation}
\end{lemma}

We note that $\vw_{\SW}$ does not, in general, equal $\vwst$. One reason for that is disparities in feature modification costs: even if a unit modification of feature $i$ leads to a high level of improvement $\vwst(i)$, this feature may be too costly to improve. Second, even when the costs are identical for all features (e.g., $A_g = I_{d \times d}$), it is still the case that $\vw_{\SW} \neq \vwst$ (unless $\Pi_1 + \Pi_2 = I_{d \times d}$), since it is often worth incentivizing feature changes in directions that overlap across and benefit \emph{both} groups, in order to maximize their joint social welfare (see example below).
\begin{example}
Let $d = 3$, and the optimal feature vector be $\vwst = (2/3,2/3,1/3)$ (note that $\|\vwst\| = 1$.). $\Pi_1$ projects to features $1$ and $3$, while $\Pi_2$ projects to features $2$ and $3$. Agents costs $I_{3 \times 3}$ in both groups, to isolate the effect of the projections on the social welfare maximizing rule. 

When posting $\vwst = (2/3,2/3,1/3)$, we have $\Delta_1(\vwst) = \Pi_1 \vwst =  (2/3,0,1/3)$ and $\Delta_2(\vwst) = \Pi_2 \vwst = (0,2/3,1/3)$; this leads to an increase in the social welfare across groups of $(\vwst)^\top \left(\Delta_1(w) + \Delta_2(w) \right) = (4/9 + 0 + 1/9) + (0 + 4/9 + 1/9) = 10/9$. 

An alternative is to put more weight on shared feature $3$, even though it yields the lowest level of improvement in each group. For example, let us pick $\vw = \frac{1}{\sqrt{3}} \cdot (1,1,1)$. We now get a better expected improvement across groups of $(\vwst)^\top \left(\Delta_1(\vw) + \Delta_2(\vw) \right) = \frac{1}{\sqrt{3}} \left((2/3 + 1/3) + (2/3 + 1/3) \right) = 2/\sqrt{3} > 10/9$.
\end{example}

Using the same techniques as for Lemma~\ref{lem:sw-opt-quasi}, we also characterize the scoring rule that maximizes the social welfare of a single group $g$. We use this as a benchmark to understand how far from optimal $\vw_{SW}$ can be within each group.

\begin{lemma}\label{lem:group-opt-vw}
The scoring rule maximizing the social welfare of group $g$ is: $\vw_g = \frac{ (A_g^{-1} \Pi_g )^\top \vwst}{\|(A_g^{-1} \Pi_g )^\top \vwst \|}$.
\end{lemma}

\section{Equilibrium Analysis}\label{sec:characterizations}

In this section, we study the societal impact of the equilibrium strategies of the principal and the agents {computed in Section \ref{sec:equilibrium-compute}}. We do so by examining feasibility of the objectives of cross-group improvement introduced in Section~\ref{sec:model}. We then study the ability to achieve improvement across groups while maintaining high predictive accuracy. We assume $(A_1^{-1} \Pi_1 + A_2^{-1} \Pi_2)^\top \vwst \neq 0$, as otherwise the objective of Eq.~\eqref{eq:vw-quasi} is always $0$. The proofs for this section can be found in Appendix \ref{app:characterizations}.

\subsection{Do-No-Harm}\label{subsec:do-no-harm}

When a benevolent\footnote{In Sec.~\ref{subsec:joint}, we instead consider a principal who wishes to trade off social welfare and predictive accuracy.} principal deploys an equilibrium rule maximizing the social welfare of the population, one could expect that this rule does not cause any negative externality (i.e., outcome deterioration). However, this is not the case in general, as we observe in the following example.

\begin{example} \label{exa:externality}
Assume that the cost and the projection matrices for the two groups are:\\ 
\[
    A_1 = 
    \begin{bmatrix}
    1 & 2 \\
    2 & 5
    \end{bmatrix}
    , 
    A_2 = 
    \begin{bmatrix}
    \phantom{-}4 & -4 \\
    - 4 & \phantom{-}8
    \end{bmatrix}
    ,\quad
    \Pi_1 = 
    \begin{bmatrix}
    \phantom{-}1/2  & -1/2 \\
    -1/2 & \phantom{-}1/2
    \end{bmatrix},
    \Pi_2 =
    \begin{bmatrix}
    1 & 0 \\
    0 & 0
    \end{bmatrix}.
    \]

Note that $A_1, A_2$ are symmetric and \PD, as their eigenvalues are $\lambda^1 = 3 \pm 2\sqrt{2}$ and $\lambda^2 = 2(3 \pm \sqrt{5})$ respectively. Further, $\Pi_1, \Pi_2$ are orthogonal projections, as $\Pi_g^2 = \Pi_g = \Pi_g^\top$. Finally, assume that ${\vwst}^\top = [0 \; \sqrt{a}]$, for scalar $a > 0$. Then, for the numerator of $\calI_2(\vwst)$ we have that: 
\begin{align*}
&{\vwst}^\top \left(A_1^{-1} \Pi_1 \Pi_2 {A_2^{-1}}^\top + A_2^{-1} \Pi_2 {A_2^{-1}}^\top \right)\vwst \\ 
&= 
{\vwst}^\top \begin{bmatrix}
\phantom{+}2 & \phantom{+}1 \\
-\nicefrac{5}{8} & -\nicefrac{5}{16}
\end{bmatrix}
\vwst = -\frac{5}{16}a < 0 \quad(a > 0).
\end{align*}
\end{example}

\begin{remark} \label{rem:best-case}
Example \ref{exa:externality} highlights the fact that, perhaps surprisingly, even assuming the ``best-case'', where the principal optimizes social welfare, is not sufficient to overcome the tension that stems from cost disparities across groups. Our experiments (App. \ref{app:experiments}) validate this counter-intuitive insight.
\end{remark}

We hence next abstract away from cost disparities (and consider cost functions that differ among groups only by a multiplicative factor\footnote{This covers most of the cost functions considered in prior work~\citep{hardtetal,dongetal,grinder,StrategicPerceptron}, where the cost matrices are diagonal with identical coefficients for all agents.}), as we wish to examine cases when discrepancy between the two groups is only due to disparities in information regarding the principal's assessment rule. We first, however, state the more general necessary and sufficient conditions for guaranteeing no negative externality.
\begin{theorem} \label{thm:do-no-harm}
In equilibrium, there is no negative externality for group $g$ {and any $\vwst$} if and only the matrix {$\left(A_1^{-1} \Pi_1+ A_2^{-1} \Pi_2\right) \Pi_g A_g^{-1} + A_g^{-1} \Pi_g \left(\Pi_1 A_1^{-1} + \Pi_2 A_2^{-1}\right)$} is $\PSD$. 
\end{theorem} 

As we show next (Corollary \ref{prop:do-no-harm1}), assuming \emph{proportional} costs between groups in fact suffices to guarantee no negative externality in any of the groups in equilibrium, \emph{regardless} of information discrepancy between them.

\begin{corollary}\label{prop:do-no-harm1}
There is no negative externality for either group in equilibrium if the cost matrices are proportional to each other; i.e., $A_1 = c \cdot A_2$ for a scalar $c > 0$.
\end{corollary}
Another interesting implication of Theorem \ref{thm:do-no-harm} is that no negative externality is experienced in equilibrium when subspaces $\calS_1, \calS_2$ are orthogonal. Intuitively, this happens because the two groups have {no} informational overlap, and hence optimal social welfare by the principal is achieved by a rule which only has to take into account a single group in each informational subspace.

\begin{corollary}\label{prop:do-no-harm2}
There is no negative externality in equilibrium, if subspaces $\calS_1, \calS_2$ are orthogonal.
\end{corollary}

Theorem \ref{thm:do-no-harm} offers an important takeaway. Namely, that information discrepancy by itself is not sufficient to cause outcome deterioration. It may, however, still result in disparities in improvement, which we discuss next.

\subsection{Equal Improvement Across groups}\label{subsec:total-outcome-improvement}

While highly desirable in itself, the ability to induce improvement simultaneously in all groups does not prevent differences in the extent of such improvements across groups. In this subsection, we hence study a stronger objective: equal improvement across groups.\footnote{Equality of total improvement is a strictly stronger objective than do-no-harm. Indeed, achieving equal total improvement guarantees that there exists no negative externality, since the optimal social welfare is always non-negative.} To isolate the effects of information discrepancy, we assume throughout it that $A_1 = A_2 = I_{d \times d}$. We first introduce a measure that will be helpful in quantifying the difference in improvement:

\begin{definition}\label{def:proxy} Given a scoring rule $\vw\in \mathbb{R}^d$ and projections $\Pi_1,\Pi_2 \in \mathbb{R}^{d\times d}$,
we define the \emph{information overlap proxy} between groups $G_1,G_2$ with respect to $\vw$ to be $r_{1,2}(\vw) := \| \Pi_1 \vw - \Pi_2 \vw \|_2$.
\end{definition}

The following lemma shows that the information overlap proxy with respect to the underlying scoring rule $\vwst$ upper bounds the difference in improvement between groups.

\begin{lemma}\label{lem:overlap-proxy}
Let $\diff_{1,2}(\vw) \triangleq |\calI_1(\vw) - \calI_2(\vw)|$ be the disparity in improvement across groups when the principal's rule is $\vw$. In equilibrium, if $A_1 = A_2 = I_{d \times d}$, then: $\diff_{1,2}(\vw_{\SW}) \leq r_{1,2} (\vwst)$. Further, equality holds if and only if $\Pi_1 \vwst$ and $\Pi_2 \vwst$ are co-linear.
\end{lemma}

\begin{remark} \label{rem:tightness}
In particular, note that the bound is tight when $\Pi_1 = \Pi_2$ (perfect overlap) and $\Pi_1 = I_{d \times d}$, $\Pi_2 = 0$ (maximum informational disparities across groups).
\end{remark}

We next derive necessary and sufficient conditions for equality of improvement in the general case.

\begin{theorem} \label{thm:eq-imp}
In equilibrium, groups have equal improvement for all $\vwst$ if and only if $A_1^{-1}\Pi_1 A_1^{-1} = A_2^{-1}\Pi_2 A_2^{-1}$.
\end{theorem}

Note that Theorem \ref{thm:eq-imp} holds globally (regardless of $\vwst$), identifying the condition for improvement disparity to \emph{vanish}. Lemma \ref{lem:overlap-proxy}, however, provides an instance-specific (as a function of  $\vwst$) upper bound. As we allow for such instance-specific analysis, equal improvement may arise under weaker conditions.

\subsection{Efficient Effort Exertion Across groups}\label{subsec:per-unit-outcome-improvement}

Achieving equality of improvement across groups does not, however, guarantee that effort to improve is exerted at the same level of efficiency across groups. Our notion of per-unit improvement (Definition \ref{def:per-unit-outcome-improve}) aims to capture the concept of optimal effort exertion formally (Remark \ref{rem:per-unit}).
This section studies the ability to ensure efficient effort exertion across groups. We begin by exhibiting the difference between improvement and efficient effort exertion.

\begin{proposition}\label{prop:rel-improvement-arbitrary}
Let $\alpha > 0$ be arbitrarily small. In equilibrium we may see simultaneously:
\squishlist
     \item arbitrarily different improvement across groups: $\calI_1(\vw_{\SW}) < \alpha \cdot \calI_2(\vw_{\SW}).$
    \item optimal per-unit improvement in all groups, i.e., $\ucalI_g(\vw_{\SW}) =  \ucalI_g(\vw_g), \forall g.$
\squishend
\end{proposition}

Next, we derive the necessary and sufficient conditions for optimal per-unit improvement in the general case.

\begin{theorem} \label{thm:suf-and-nec-per-unit} In equilibrium, group $g$ gets optimal per-unit {improvement} if and only if:
\[
\Bigg\langle A_g^{-1}\frac{\Pi_g A_g^{-1}  \vwst}{\left\Vert\Pi_g A_g^{-1} \vwst \right\Vert_2}
- A_g^{-1}\frac{\Pi_g\left( \Pi_1 A_1^{-1} +  \Pi_2 A_2^{-1} \right) \vwst}{\left\Vert\Pi_g\left(\Pi_1 A_1^{-1}  + \Pi_2 A_2^{-1} \right) \vwst \right\Vert_2},\vwst\Bigg\rangle = 0.
\]
\end{theorem} 

Since the condition in Theorem \ref{thm:suf-and-nec-per-unit} can be difficult to interpret, we identify two natural cases where optimal per-unit improvement is guaranteed. The first case occurs when their groups' information on the decision rule does not overlap.

\begin{corollary}\label{cor:optimal-per-unit-1}
In equilibrium, optimal per-unit improvement across groups is guaranteed if $\calS_1, \calS_2$ are orthogonal.
\end{corollary}

The second case is when groups have the same information regarding the decision rule, and their feature modification costs are proportional to one another.

\begin{corollary} \label{cor:optimal-per-unit-2}
In equilibrium, optimal per-unit improvement across groups is guaranteed when the cost matrices are proportional to each other and $\Pi_1=\Pi_2$.
\end{corollary}

Intuitively, both Corollary \ref{cor:optimal-per-unit-1} and Corollary \ref{cor:optimal-per-unit-2} reflect situations where the \emph{direction} of the optimized solution in each of the groups' informational subspace is not affected by the other groups.

\subsection{Improvement With High Predictive Accuracy}\label{subsec:joint}

In this subsection, we replace the assumption of a benevolent principal with one that wishes to trade-off predictive accuracy and social welfare. In other words, we study the ability to induce improvement in all groups while deploying highly-accurate decision rules. The proofs of this subsection and additional intuition can be found in Appendices \ref{app:joint} and \ref{app:commuting}.

A simple way to take into account both accuracy and social welfare objectives is to consider decision rules of the form $\lambda \vwst + (1-\lambda) \vw_{\SW}$ for $\lambda \in [0,1]$. Such rules exhibit a trade-off between picking the accuracy-optimizing (as $\lambda \to 1$) and the welfare-optimizing (as $\lambda \to 0$) rules. We investigate conditions under which a decision rule of this form can induce improvement in all groups. To do so, we begin by introducing a simple condition regarding modification costs.

\begin{definition} \label{def:decomposable} We say that group $g$ has decomposable modification cost, if,
for all $\Delta_x \in S_g$, $A_g \Delta_x \in S_g$, and for all $\Delta_x' \in S_g^\bot$, $A_g \Delta_x' \in S_g^\bot$.
\end{definition}

At a high level, the condition in Definition \ref{def:decomposable} requires the cost of any modification of features to be decomposable into two \emph{independent} components: the cost of modifications within the group's subspace of information, and the cost of modifications outside of it.\footnote{We further note that a similar condition arises in the context of the principal's learning problem (appendix \ref{app:learning}).} We refer the reader to Appendix \ref{app:commuting} for more intuition regarding why such condition may arise naturally, but note that this condition encodes that agents in group $g$ \emph{never} perform manipulations outside of their informational subspace $S_g$, i.e. modifications whose effect they have no understanding of. 

As we prove next, the condition in Definition \ref{def:decomposable} ensures improvement in all groups while maintaining high predictive accuracy.

\begin{theorem} \label{thm:joint}
Assume group $g$ has decomposable modification cost. Then, Do-No-Harm for group $g$ is guaranteed if the principal deploys $\vwst$. Further, if Do-No-Harm is guaranteed for group $g$ under $\vw_{\SW}$, it is also guaranteed for any convex combination of $\vw_{\SW}$ and $\vwst$.
\end{theorem}

In particular, Theorem \ref{thm:joint} shows that even when the principal optimizes for accuracy \emph{alone}, no negative externality is experienced in any of the groups under the condition of Definition \ref{def:decomposable}. We next show a surprising implication of Theorem \ref{thm:joint}. Namely, that unlike the case for the social welfare maximizing solution (as shown in Example \ref{exa:externality}), improvement in all groups may in fact be naturally guaranteed for the accuracy-maximizing solution.

\begin{corollary} \label{cor:full}
Under full information ($\Pi_1=\Pi_2=\mathbbm{I}_{d \times d}$), Do-No-Harm for all groups is guaranteed for $\vwst$.
\end{corollary}

Corollary \ref{cor:full} highlights an interesting perspective for the benefits of transparency in prediction in strategic contexts; even a mostly self-interested principal could in fact benefit from making its rule more transparent. For example, one could consider this in the context of loan approvals, where a bank deploys a proprietary decision rule, aiming primarily for high predictive accuracy, and secondarily for increasing the quality of loan candidates across all groups. Corollary \ref{cor:full} can then be viewed as an incentive for the bank to increase the transparency of such rule.

\section{Experiments}\label{sec:experiments}

Here, we empirically evaluate the impact of information disparities at equilibrium on two real-world datasets that pertain to our setting: the $\textsc{Taiwan-Credit}$ and $\textsc{Adult}$ datasets.\footnote{Available at \url{https://archive.ics.uci.edu/ml/datasets/default+of+credit+card+clients} \& \url{https://archive.ics.uci.edu/ml/datasets/adult}.} Our code is available in the supplementary.

\xhdr{Experimental Setup.}
For $\textsc{Taiwan-Credit}$ $d = 24$ and $\textsc{Adult}$ $d = 14$. In order to guarantee \emph{numerical} (rather than categorical) feature values, we pre-processed the $\textsc{Adult}$ dataset to transform the categorical ones to integers. Specifically, for the features for which there was a \emph{clear} hierarchical ordering (e.g., the ``Education'' feature, where we could order agents in terms of their highest education level reached), we reflected this ordering in the assignment of numerical values to these categories. For the $\textsc{Taiwan-Credit}$ dataset, no pre-processing was needed.

\begin{table}[htbp]
\caption{\vspace{-10pt}Groups for the $\textsc{Taiwan-Credit}$ dataset.}
\begin{center}
\resizebox{0.7\columnwidth}{!}{%
\begin{tabular}{ |c|c|c|c|c| } 
 \hline
   {}       & {\bf Age}       & {\bf Education } & \bf{Marriage} \\ \hline
 {$G_1$}    & $\leq25$ yrs old  & gradschool \& college     & married  \\ \hline
 {$G_2$}    & $>25$ yrs old     & high school    & not-married \\\hline
 {\bf {$\diff(\vw_{\SW})$}}  & $0.34$      & $0.05$      &  $0.23$  \\ \hline
 {\bf {$r_{1,2}(\vwst)$}}  & $0.5$      & $0.52$      &  $0.48$  \\ \hline
\end{tabular}
}
\end{center}
\label{table:dataset1}
\end{table}

\begin{table}[htbp]
\caption{\vspace{-10pt}Groups for the $\textsc{Adult}$ dataset.}
\begin{center}
\resizebox{0.7\columnwidth}{!}{%
\begin{tabular}{ |c|c|c|c|c| } 
 \hline
   {}       & {\bf Age}       & {\bf Country } & {\bf Education}\\ \hline
 {$G_1$}    & $\leq 35$ yrs old  & western world     &  degrees $\geq$ high school     \\ \hline
 {$G_2$}    & $>35$ yrs old     & everyone else     &  degrees $<$ high-school    \\\hline
  {\bf {$\diff(\vw_{\SW})$}}  & $0.15$      & $0.66$      &   $0.20$    \\ \hline
 {\bf {$r_{1,2}(\vwst)$}}  & $0.29$      & $0.89$      &   $0.77$     \\ \hline
\end{tabular}
}
\end{center}
\label{table:dataset2} 
\end{table}

In both cases, we ran ERM in order to identify $\vwst$ and we assumed that costs are $A_1 = A_2 = I_{d \times d}$. In App.~\ref{app:experiments} we present additional experimental results for cost matrices $A_1, A_2$ that differ from one another. After the pre-processing step, we created the groups of the population based on categories that intuitively ``define'' peer-networks. Our judgment for picking these categories is based on folklore ideas about how people choose their network and social circles. For the $\textsc{Taiwan-Credit}$ dataset, we use the following categories: Age, Education, and Marriage, while for the $\textsc{Adult}$ dataset, we use: Age, Country, Education, and the final groups are in Tables~\ref{table:dataset1} and~\ref{table:dataset2}. In order to obtain the projection matrices $\Pi_1, \Pi_2$, we ran $\texttt{SVD}$ on the points inside of $G_1, G_2$. To be more precise, let $X_i \in \bbR^{|G_i| \times d}$ be the matrix having as rows the vectors $\vx^\top, \forall \vx \in G_i$. Then, running $\texttt{SVD}$ on $X$ produces three matrices: $X_g = U \; D \; V_g^\top$, where $V \in \bbR^{d \times r}$ and $r = \texttt{rank}(X_g)$. Let $V_{g,5}$ correspond to the matrix having as columns the eigenvectors corresponding to the $5$ top eigenvalues and zeroed out all other $d - 5$ columns. Then, the projection matrix $\Pi_g$ is defined as $\Pi_g = V_{g,5} V_{g,5}^\top$.\footnote{{Effectively, we focus the feature space on the directions corresponding the top $5$ eigenvalues found in each group's data, as per traditional principal component analysis.}}

\xhdr{Results.} In summary, our experimental results validate our theoretical analysis, and extend our insights to when the projection matrices do not satisfy the exact conditions required by the formal statements of Sec.~\ref{sec:characterizations}. 

First, we see that in both datasets, the principal's rule that optimizes the social welfare does not cause \emph{any} negative externality when $A_1 = A_2 = I_{d \times d}$ (that said, we do observe outcome deteriorations when the cost matrices differ from one another -- see Appendix~\ref{app:experiments}). In fact, we observe \emph{strict} improvement, i.e., $\calI_g(\vw_{\SW}) > 0$ for all groups $g$. 

Second, neither the total nor the per-unit improvements are equal. In terms of total improvements, we in fact see \emph{significant} disparities when the groups are defined based on their Age or their Marital Status in the \textsc{Taiwan-Credit} dataset and based on every categorization in the \textsc{Adult} dataset. These significant disparities for the particular groups we created match our intuition: we expect that people from significantly different age groups or countries to have very different understandings of the scoring rule, in turn leading to possibly very disparate total improvements. 

We note also that in both datasets the difference in the total improvements of the groups is upper bounded by the overlap proxy (i.e., $\diff(\vw_{\SW}) \leq r_{1,2}(\vwst)$), as expected from Lemma~\ref{lem:overlap-proxy}. That said, the gap between the two quantities can be rather large. This is because the magnitude of the overlap is \emph{not} the only factor controlling $\diff(\vw_{\SW})$. Rather, other factors (e.g., the direction of the overlap or how it compares to $\vwst$) also matter significantly.

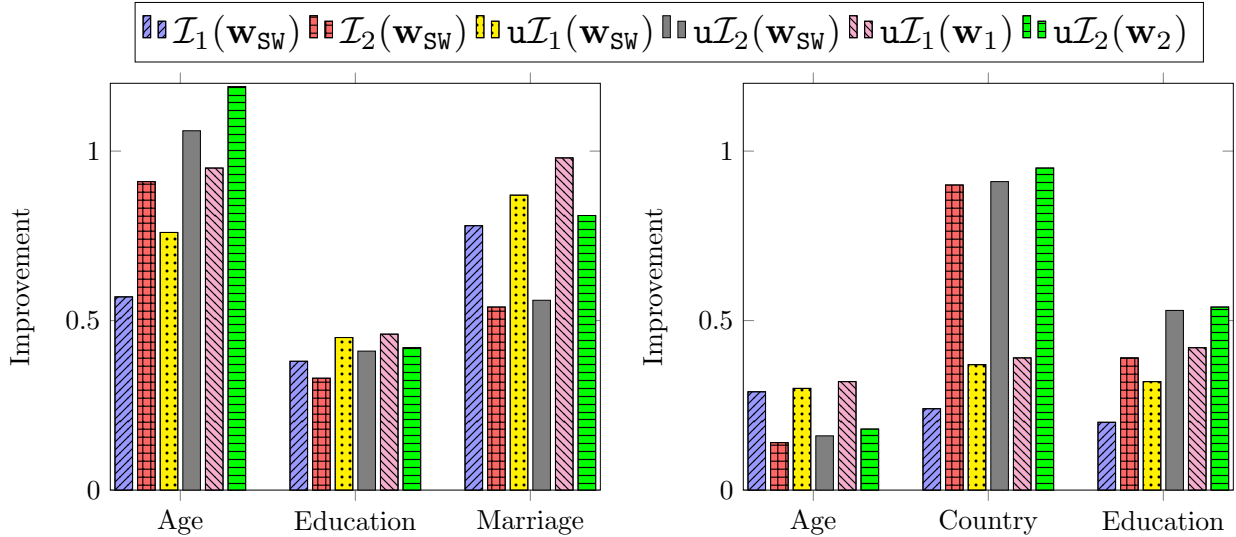
\begin{figure}[htbp]
\centering
\resizebox{\linewidth}{!}{%
\begin{tikzpicture}
        
        \begin{axis}[
        name=ax1,
        ybar,
        enlarge x limits = 0.2,
        bar width=7pt,
        ylabel={Improvement},
        legend style={font=\small, at={(-0.1,0.-0.2)},anchor=north west,legend columns=2},
        symbolic x coords={Age, Education, Marriage},
        ymax=1.2,
        ymin=0,
        xtick=data,
        ytick={0, 0.5, 1},
        ]

\addlegendentry{$\calI_1(\vw)$}
\addplot[black,fill=blue!40!white,postaction={pattern=north east lines}] plot coordinates {
        (Age,       0.57)
        (Education, 0.38)
        (Marriage,  0.78)
        
};

\addlegendentry{$\calI_2(\vw)$}
\addplot [black, fill=red!60!white, postaction={pattern=grid}]
plot coordinates {
        (Age,       0.91)
        (Education, 0.33)
        (Marriage,  0.54)
};

 \addlegendentry{$\ucalI_1(\vw)$}
 \addplot[black, fill=yellow, postaction={pattern=dots}] plot coordinates {
         (Age,       0.76)
         (Education, 0.45)
         (Marriage,  0.87)
 };

 \addlegendentry{$\ucalI_2(\vw)$}
 \addplot[black, fill=gray] plot coordinates {
         (Age,       1.06)
         (Education, 0.41)
         (Marriage,  0.56)
 };

 \addlegendentry{$\ucalI_1({\vw_1})$}
 \addplot[black, fill=magenta!40!white, postaction={pattern=north west lines}] plot coordinates {
         (Age,       0.95)
         (Education, 0.46)
         (Marriage,  0.98)
 };

 \addlegendentry{$\ucalI_2^(\vw_2)$}
\addplot [black, fill=green, postaction={pattern=horizontal lines}]
 plot coordinates {
         (Age,       1.19)
         (Education, 0.42)
         (Marriage,  0.81)
 };

\legend{}

\end{axis}

    \begin{axis}[
        at={(ax1.south east)},
        xshift=2cm,
        ybar,
        enlarge x limits = 0.2,
        bar width=7pt,
        ylabel={Improvement},
        legend style={font=\small, nodes={scale=1.5, transform shape}, at={(0.94,1.2)},anchor=north east,legend columns=6},         symbolic x coords={Age, Country, Education, Race},
        ymax=1.2,
        ymin=0,
        xtick=data,
        ytick={0, 0.5, 1},
        ]

\addlegendentry{$\calI_1(\vw_{\SW})$}
\addplot[black,fill=blue!40!white,postaction={pattern=north east lines}] plot coordinates {
        (Age,           0.29)
        (Country,       0.24)
        (Education,     0.20)
};

\addlegendentry{$\calI_2(\vw_{\SW})$}
\addplot [black, fill=red!60!white, postaction={pattern=grid}]
plot coordinates {
        (Age,       0.14)
        (Country,   0.9)
        (Education, 0.39)
};

 \addlegendentry{$\ucalI_1(\vw_{\SW})$}
 \addplot[black, fill=yellow, postaction={pattern=dots}] plot coordinates {
         (Age,       0.30)
         (Country,   0.37)
         (Education, 0.32)
 };

 \addlegendentry{$\ucalI_2(\vw_{\SW})$}
 \addplot[black, fill=gray] plot coordinates {
         (Age,       0.16)
         (Country,   0.91)
         (Education, 0.53)
 };

 \addlegendentry{$\ucalI_1(\vw_1)$}
 \addplot[black, fill=magenta!40!white, postaction={pattern=north west lines}] plot coordinates {
         (Age,       0.32)
         (Country,   0.39)
         (Education, 0.42)
 };

\addlegendentry{$\ucalI_2(\vw_2)$}
\addplot [black, fill=green, postaction={pattern=horizontal lines}]
 plot coordinates {
         (Age,       0.18)
         (Country,   0.95)
         (Education, 0.54)
 };

\end{axis}
\end{tikzpicture}
}
\caption{Left, Right: evaluation on the $\textsc{Taiwan-Credit}$ and $\textsc{Adult}$ dataset respectively. Tables~\ref{table:dataset1} and~\ref{table:dataset2} contain the characteristics of groups $G_1, G_2$. Recall that $\calI_g(\vw_{\SW}), \ucalI_g(\vw_{\SW})$, denote the total and per-unit improvement for group $g$ in equilibrium respectively, while $\ucalI_g(\vw_g)$ denotes the optimal per-unit improvement for group $g$.
\vspace{-5pt}}\label{fig:experiments}
\end{figure}

Note that \emph{optimal per-unit} improvement can be very different across groups; an extreme example is the Country categorization in \textsc{Adult}. It is surprising, however, that for the Education categorization in \textsc{Taiwan-Credit} the optimal per-unit improvement is almost identical. Another interesting observation is that the \emph{per-unit} improvement is close to (or almost the same as!) the optimal per-unit improvement for all groups in \textsc{Adult}. We suspect that this is due to having very different projection matrices $\Pi_1, \Pi_2$. 

\section{Conclusion}\label{sec:conclusion}

In this work, we have taken a first step towards understanding the implications of inaccessible decision rules for different sub-populations in strategic learning. Our results establish a close connection between the informational overlap across sub-populations, the extent to which it is possible to ensure improvement in each, and whether such improvement can be induced while maintaining high predictive accuracy. 

We discuss next limitations of our work and avenues for future research. First, our model incorporates a linearity assumption regarding the form of the decision rule; we believe this linear assumption is a natural and simple choice for a first model that studies the phenomenon of information discrepancy, but an interesting future direction would be to understand which of our insights translate to a non-linear model, and what new insights arise. Second, we assume agents best-respond perfectly to the principal's choices. While one can weaken the best-response assumption, this may affect the sharpness of our results. We note that some form of assumption regarding how agents respond to the model is natural and reflects many real-life situations. Finally, while we provide guarantees for improvement with high predictive accuracy in the form of safeguarding against outcome deterioration, it would be interesting to further study the trade-offs between accuracy and improvement.

\section{Acknowledgments}
Yahav Bechavod was supported in part by Israel Science Foundation (ISF) grants 1044/16 and 2861/20, the United States Air Force and DARPA under contract FA8750-19-2-0222, and the Apple Scholars in AI/ML PhD Fellowship. Chara Podimata was supported by a Microsoft Dissertation Grant and a Siebel scholarship. Juba Ziani was funded in part by the Warren Center for Network and Data Sciences at the University of Pennsylvania, and NSF grant AF-1763307. Any opinions, findings and conclusions or recommendations expressed in this material are those of the author(s) and do not necessarily reflect the views of the United States Air Force and DARPA. 

\newpage
\bibliographystyle{plainnat}
\bibliography{refs.bib}

\newpage
\appendix
\section{The Principal's Learning Problem}\label{app:learning}

Up until know, we have assumed that the principal has full information on the parameters of the problem. In particular, the principal perfectly knows the underlying linear model $\vwst$, the cost matrices $A_1$ and $A_2$, and the projection matrices $\Pi_1$ and $\Pi_2$. In this section, we study how our principal can learn $\vw_{\SW}$ from samples of agents' \emph{modified} features.

To do so, we present two simple building blocks: one that uses a batch of observations to help us estimate $\vwst$, and one that, aims to estimate $\Delta_g(\vw) = A_g^{-1} \Pi_g \vw$ for a given $\vw$. We make the following commutativity assumption (please see Appendix \ref{app:commuting} for more regarding this assumption):
\begin{assumption}\label{as:commute}
For all $g \in \{1,2\}$, $\Pi_g A_g^{-1} = A_g^{-1} \Pi_g$. 
\end{assumption}
We remark that this assumption holds in several cases of interest. For example, this holds when $A_g = \sigma_g \mathbb{I}_{d \times d}$ for some $\sigma_g \geq 0$, i.e. when the cost of an agent for modifying features is the same across all features and independent across features. This also happens when $\Pi_g$ and $A_g^{-1}$ are both diagonal, in which case they are simultaneously diagonalizable hence commute (for example, when $\Pi_g$ is the projection to a subset of the features, and when manipulating one feature does not affect another feature for free).

Under Assumption~\ref{as:commute}, Equation~\eqref{eq:principal-vw2} can be rewritten as
\[
\vw_{\SW} = \frac{ A_1^{-1} \Pi_1 \vwst + A_2^{-1} \Pi_2 \vwst}{\Vert A_1^{-1} \Pi_1 \vwst + A_2^{-1} \Pi_2 \vwst \Vert}
= \frac{\Delta_1(\vwst) + \Delta_2(\vwst)}{\Vert \Delta_1(\vwst) + \Delta_2(\vwst) \Vert},
\]
Accurate estimation of both $\Pi_g \vwst$ and $\Delta_g(\vw)$ for any given $\vw$ is sufficient for accurate estimation of $\vw_{\SW}$. The principal can then take a classical explore-then-exploit approach, in which she first sets aside a batch of agents in group $g$ to estimate the parameters of the problem to her desired accuracy, then use the parameters she learned to incentivize optimal outcome improvement during the rest of the time horizon. 

\xhdr{Estimating $\mathbf{\Pi_g \vwst}$} To estimate $\Pi_g \vwst$, we use Algorithm~\ref{algo:learning_w}. The algorithm has access to $n$ agents from group $g$. It consists of first posting an initial model of $\vw = 0$ w.l.o.g.\footnote{
The choice of $\vw = 0$ in Algorithm~\ref{algo:learning_w} is not crucial. In fact, picking any given $\vw$ induces the same distribution of feature vectors as $\calS_g$, with its expectation shifted by a constant amount of $\Delta_g(\vw)$, after the first $N_g$ unmodified observations. In turn, given $N_g + n$ samples, the distribution of the last $n$ feature vectors used for estimation remains full-rank in subspace $\calS_g$ and still has covariance matrix $\Sigma_g$. Therefore, the high-probability bound of Claim~\ref{clm:vstar_cctr} remains the same. In most cases, $N_g$ is of the order of the dimension $d$ of the problem, hence we have that $n >> N_g$, and the cost of waiting for agent to learn $\vw$ is minimal.}
, observing the agents' true, unmodified features and true labels (according to $\vwst$), and using these observations to compute and output an estimate $\bar{\vw}$ of $\Pi_g \vwst$:

\begin{algorithm}[htbp]
\caption{Estimating $\Pi_g \vwst$}\label{algo:learning_w}
\begin{algorithmic}
\STATE Post $\vw = 0$;
\FOR{$i=1$ {\bfseries to} $n$}
\STATE Principal observes agent $i$'s true feature vector $\vx_i$, and his true label $y_i$;
\ENDFOR
\STATE Output $\bar{\vw} \triangleq \arg \min_{\vw} \sum_{i=1}^n \left(\vx_i^\top \Pi_g \vw - y_i\right)^2$;
\end{algorithmic}
\end{algorithm}

For simplicity of exposition, we consider the case in which the noise in the label follows a Gaussian distribution, as per the below assumptions. We note however that our results classically extend to the sub-Gaussian case by classical recovery guarantees of linear least-square regression.
\begin{assumption}\label{as:gaussian_noise}
For every agent $i$, $y_i - \vx_i^\top \vwst \sim \mathcal{N}(0,\sigma^2)$ where $0 \leq \sigma^2 < \infty$.
\end{assumption}

\begin{claim}\label{clm:vstar_cctr}
Under Assumption~\ref{as:gaussian_noise}, with probability at least $1-\delta$, the output $\bar{\vw}$ of Algorithm~\ref{algo:learning_w} satisfies
\begin{align*}
    \left \Vert \bar{\vw} - \Pi_g \vwst \right \Vert_2 
    = O\left(\frac{\sigma^2 d \log(1/\delta)}{\lambda_g n}\right),
\end{align*}
where $\lambda_g$ denotes the smallest non-zero eigenvalue of $\Sigma_g$, the covariance matrix of distribution $\calD_g$.
\end{claim}

\begin{proof}
Without loss of generality, we restrict attention to the subspace $\calS_g$ induced by the support of the distribution of features $\calD_g$. Let $\Sigma_g$ be the covariance matrix of distribution $\calD_g$; by definition of $\calD_g$, $\Sigma_g$ is full-rank with smallest eigenvalue $\lambda_g$ in $\calS_g$. By the classical recovery results on least-square regression, since $\E \left[y_i | \vx_i\right] = \vx_i^\top (\Pi_g \vwst)$ by assumption, we obtain that 
\begin{align*}
    \Vert \bar{\vw} - \Pi_g \vwst\Vert_2 = O\left(\frac{\sigma^2 d \log(1/\delta)}{\lambda_g n}\right).
\end{align*}
This concludes the proof.
\end{proof}

\xhdr{Estimating $\Delta_g(\vw)$.} Algorithm~\ref{algo:learning_delta} has access to $2n + N_g$ agents from group $g$, takes as an input a vector $\vw$, and outputs an estimate of $\Delta_g(\vw)$.

\begin{algorithm}[htbp]
\caption{Estimating $\Delta_g(\vw)$}\label{algo:learning_delta}
\begin{algorithmic}
\STATE Post $\vw_1 = 0$\;
\FOR{$i=1$ {\bfseries to} $n$}
\STATE Principal observes agent $i$'s true feature vector $\vx_i$, and his true label $y_i$\;
\ENDFOR
\STATE Post $\vw_2 = \vw$\;
\FOR{$i = n + 1$ {\bfseries to} $n + N_g$}
\STATE agent $i$ plays true feature vector $\vx_i$\;
\ENDFOR
\FOR{$i = n + N_g + 1$ {\bfseries to} $2 n + N_g$}
\STATE principal observes agent $i$'s modified feature vector $\hat{\vx}_i = \vx_i + \Delta_g(\vw)$\;
\ENDFOR
\STATE Output $\overline{\Delta_g} \triangleq \frac{1}{n} \left(\sum_{i= n + N_g +1}^{2n + N_g}  \hat{\vx}_i - \sum_{i=1}^{n} \hat{\vx}_i \right)$
\;
\end{algorithmic}
\end{algorithm}

\begin{claim}
Let us assume that for all $i$, $\Vert x_i \Vert_{\infty} \leq 1$. Then, with probability at least $1-\delta$, the output $\overline{\Delta_g}$ of Algorithm~\ref{algo:learning_delta} satisfies
\[
\left \Vert \overline{\Delta_g} - \Delta_g(\vw) \right \Vert_2 
\leq \sqrt{\frac{d \log(d/2\delta)}{n}}.
\]
\end{claim}

\begin{proof}
First, we note that 
\begin{align*}
\overline{\Delta_g} 
\triangleq \frac{1}{n} \left(\sum_{i= n+ N_g + 1}^{2n + N_g} \hat{\vx}_i - \sum_{i=1}^{n} \hat{\vx}_i \right)
& = \frac{1}{n}\left(\sum_{i= n+ N_g + 1}^{2n + N_g} \vx_i + n \Delta_g(\vw) - \sum_{i=1}^{n} \vx_i + \right)
\\&= \Delta_g(\vw) + \frac{1}{n} \left(\sum_{i= n+ N_g + 1}^{2n + N_g} \vx_i - \sum_{i=1}^{n} \vx_i + \right).
\end{align*}
In turn, we have that 
\[
\left \Vert \overline{\Delta_g} - \Delta_g(\vw) \right \Vert_2
= \frac{1}{n} \left \Vert \sum_{i= n+ N_g + 1}^{2n + N_g} \vx_i - \sum_{i=1}^{n} \vx_i \right \Vert_2.
\]
We have that 
\[
\sum_{i= n+ N_g + 1}^{2n + N_g} \vx_i - \sum_{i=1}^{n} \vx_i = \sum_{i=1}^n Z_i
\]
where $Z_i = \vx_{i+n+N_g} - \vx_{i}$. In turn, $Z_i$ is a random vector with mean $\E[Z_i] = 0$ and covariance matrix $2 \Sigma_g$, noting that $\vx_i$ and $\vx_{i+n+N_g}$ are drawn independently. Further, $\left\vert Z_i(k) \right\vert \leq 2$. By Hoeffding's inequality, we have that with probability at least $1 - \frac{\delta}{d}$, for a given $k \in [d]$,
\[
\left \vert \sum_{i=1}^n Z_i(k) \right \vert \leq \sqrt{2 n \log(d/2\delta)}.
\]
By union bound, we have that with probability at least $1-\delta$, this holds simultaneously for all $k \in [d]$, with directly yields
\[
\left \Vert \sum_{i=1}^n Z_i \right \Vert = \sqrt{\sum_{k=1}^d \left(\sum_{i=1}^n Z_i(k)\right)^2 }
\leq \sqrt{2 d n \log(d/2\delta)}.
\]
This immediately leads to the result.
\end{proof}
\section{Supplementary Material for Section~\ref{sec:equilibrium-compute}}\label{app:equilibrium-compute}

\subsection{Omitted Proofs from Subsection \ref{subsec:surr-vw-g}}

\begin{proof}[Proof of Lemma~\ref{lem:inferred-rule}]
We first identify the rules $\tvw$ that are the solutions of the error minimization part of Eq.~\eqref{eq:erm}:
\begin{align*}
    \tvw    
    &= \arg \min_{\vw'} \sum_{i \in [N_g]} \left( \vx_{g,i}^\top \vw' - \hy_{g,i} \right)^2\\ 
    &= \arg \min_{\vw'} \underbrace{\sum_{i \in [N_g]} \left( \left(\Pi_g \vx_{g,i} \right)^\top \vw' - \hy_{g,i} \right)^2}_{f(\vw')} \numberthis{\label{eq:bef-opt}}
\end{align*}
where the second equation is due to the fact that since $\forall \vx_{g,i} \sim \calD_g$, it holds that $\Pi_g \vx_{g,i} = \vx_{g,i}$. To solve the minimization problem of Eq.~\eqref{eq:bef-opt}, we take the first order conditions, so at the optimal $\tvw$:
\begin{equation}\label{eq:first-order}
    \nabla f\left( \tvw \right) = 0 \Leftrightarrow 2 \sum_{i \in [N_g]}\left( \Pi_g \vx_{g,i} \right) \left( \left(\Pi_g \vx_{g,i} \right)^\top \tvw - \hy_{g,i} \right) = 0
\end{equation}
Now, $\tvw = \Pi_g \vw$ is one of the solutions of Eq.~\eqref{eq:first-order}, since $\hy_{g,i} = (\Pi_g \vx_{g,i})^\top \vw = (\Pi_g \vx_{g,i} )^\top \Pi_g \vw$. Next, we argue that due to the norm-minimization rule we use for tie-breaking, it is also the \emph{unique} solution. To do so, let $\tvw$ be a norm-minimizing solution of Eq.~\eqref{eq:first-order}, and write $\tvw = \Pi_g \vw + \vx'$, where $\vx'$ is an arbitrary vector; note that this is without loss of generality. We can write $\tvw = \Pi_g \vw + \Pi_g \vx' + \Pi_g^\bot \vx'$ (where $\Pi_g^\bot \vx'$ is the projection of $\vx'$ in the orthogonal subspace of $\calS_g$). Now, note that $\Pi_g \tvw = \Pi_g \vw + \Pi_g \vx' + \Pi_g \Pi_g^\bot \vx' = \Pi_g \vw + \Pi_g \vx'$ is also a solution to Eq.~\eqref{eq:first-order}, as
\begin{align*}
   &\sum_{i \in [N_g]}\left( \Pi_g \vx_{g,i} \right) \left( \left(\Pi_g \vx_{g,i} \right)^\top \Pi_g \tvw - \hy_{g,i} \right)=\\ 
   &= \sum_{i \in [N_g]}\left( \Pi_g \vx_{g,i} \right) \left( \vx_{g,i}^\top  \Pi_g^\top \Pi_g \tvw - \hy_{g,i} \right) 
   \\&= \sum_{i \in [N_g]}\left( \Pi_g \vx_{g,i} \right) \left( \vx_{g,i}^\top  \Pi_g \tvw - \hy_{g,i} \right), 
\end{align*}
where the last step is due to the fact that $\Pi_g$ represents an orthogonal projection, hence $\Pi_g^\top = \Pi_g$ and $\Pi_g \Pi_g = \Pi_g$. Further, if $\Pi_g^\bot \vx' \neq 0$, we have that
\begin{align*}
\Vert \tvw \Vert_2 
&= \Vert \Pi_g \vw + \Pi_g \vx' \Vert_2  + \Vert \Pi_g^\bot \vx' \Vert_2\\
&> \Vert \Pi_g \vw + \Pi_g \vx' \Vert_2\\
&= \Vert \Pi_g \tvw \Vert_2
\end{align*}
by orthogonality of $\Pi_g \left(\vw + \vx'\right)$ and $\Pi_g^\bot \vx'$. This contradicts the fact that $\tvw$ is a norm-minimizing solution of Eq.~\eqref{eq:first-order}. Therefore, we have $\tvw = \Pi_g \vw + \Pi_g \vx'$. 

Using this together with $\hy_{g,i} = (\Pi_g \vx_{g,i} )^\top \Pi_g \vw$, the left-hand side of Eq.~\eqref{eq:first-order} becomes: 
\begin{align*}
\sum_{i \in [N_g]} \left( \Pi_g \vx_{g,i} \right) 
\left( \left( \Pi_g \vx_{g,i}\right)^\top \Pi_g \vw + \left( \Pi_g \vx_{g,i} \right)^\top \Pi_g \vx' -\left( \Pi_g \vx_{g,i}\right)^\top \Pi_g \vw \right)
= \sum_{i \in [N_g]} \left( \Pi_g \vx_{g,i} \right) \left(\Pi_g \vx_{g,i} \right)^\top \vx' \numberthis{\label{eq:bef-claim}}
\end{align*}
where the last equality comes from the fact that $\Pi_g \Pi_g^\bot = 0_{d \times d}$. We next prove a technical lemma. 
\begin{lemma}\label{lem:intermediate}
Let $Z \triangleq \sum_{i \in [N_g]} \left( \Pi_g \vx_{g,i} \right) \left( \Pi_g \vx_{g,i} \right)^\top$ be full-rank in subspace $\calS_g$. Then, for any vector $\vx' \in \bbR^d$ it holds that $Z \left( \Pi_g \vx' \right) = 0$ \emph{if and only if} $\Pi_g \vx' = 0$. 
\end{lemma}

\begin{proof}[Proof of Lemma~\ref{lem:intermediate}]
If $\Pi_g \vx' = 0$ then, it holds that $Z (\Pi_g \vx') = 0$. So, assume that $Z \left( \Pi_g \vx' \right) = 0$.

Let $r = \texttt{rank}(\calS_g)$ (hence, $r = \texttt{rank}(Z)$). Let us denote by $v_1, \dots, v_r$ the eigenvectors of $Z$ corresponding to eigenvalues $\lambda_1, \dots, \lambda_r$ for which $\lambda_i > 0, \forall i \in [r]$; note that $(v_1,\ldots,v_r)$ span $\calS_g$. For the rest of the eigenvalues (i.e., $i \in [r+1, d], \lambda_i = 0)$ the remaining eigenvectors are denoted as $v_{r+1}, \dots, v_d$. Without loss of generality, we take $v_i$ to have norm $1$ for all $i$; since $Z$ is a symmetric matrix, $(v_1,\ldots,v_r)$ is an orthonormal basis for $\calS_G$ and $(v_1,\ldots,v_d)$ is an orthonormal basis for $\mathbb{R}^d$.

Let $V$ denote the $d \times d$ matrix $[v_1^\top \; v_2^\top \; \cdots v_d^\top ]$ that is the change of basis that transforms the standard basis into $(v_1,\ldots,v_d)$. By orthonormality of $(v_1,\ldots,v_d)$, $V$ is unitary (i.e., $V^\top V = \mathbbm{I}$). In turn, 
\begin{equation}\label{eq:before-sub}
    Z = \left( \Pi_g \vx' \right) = V^\top V \, Z \, V^\top V \, \Pi_g \vx'
\end{equation}
Let us define matrices $P_1 = V \, Z \, V^\top$ and $P_2 =  V \, \Pi_g \vx'$. $P_1$ is a diagonal matrix having $\lambda_1, \dots \lambda_d$ on the diagonal (and hence, it only has positive values until row $r$ and $0$'s for rows in $\{r+1, d \}$). Also,
\begin{align*}
P_2 = V \, \Pi_g \, \vx' = [a_1 \; \cdots \; a_r \; 0 \; \cdots \; 0]^\top, \; \text{where} \; a_i = v_i^\top \left(\Pi_g \vx'\right).
\end{align*}
Substituting the values of $P_1, P_2$ in Eq.~\eqref{eq:before-sub} we have that: 
\begin{equation*}
    Z \left( \Pi_g \vx' \right) = V^\top \; \left[\lambda_1 a_1 \; \cdots \; \lambda_r a_r \; 0 \; \cdots \; 0 \right]^\top
\end{equation*}
But $Z \left( \Pi_g \vx' \right) = 0$ if and only if $\lambda_i a_i = 0, \forall i \in [r]$, because $V$ is invertible. Since $\lambda_i > 0$ for $i \in [r]$, it must be that $a_i = 0$. Since then, $V \Pi_g \vx' = 0$ and $V$ is invertible, this implies that $\Pi_g x' = 0$. 
\end{proof}
Defining $Z$ as $Z = \sum_{i \in [N_g]} \left( \Pi_g \vx_{g,i}\right) \left(\Pi_g \vx_{g,i} \right)^\top$, \footnote{Given enough samples from the peer dataset (i.e., a large enough $N_g$), one can guarantee that $Z$ is full rank.}  then from Lemma~\ref{lem:intermediate}, Eq.~\eqref{eq:bef-claim} is equal to $0$ \emph{if and only if} $\Pi_g \vx' = 0$. This directly yields $\tvw = \Pi_g \vw + \Pi_g \vx' = \Pi_g \vw$.
\end{proof}

\subsection{Omitted Proofs from Subsection \ref{subsec:closedform}}

\begin{proof}[Proof of Lemma~\ref{lem:br-computation}]
The function in Eq.~\eqref{eq:utility-simplified} is \emph{concave}. At the optimum $\vx'$ from the first order conditions we have that $\nabla u \left( \vx, \vx' ; g \right) = \Pi_g \vw - A_g (\vx' - \vx) = 0$. Solving the latter in terms of $\vx'$ and using the fact that matrix $A_g$ is positive definite (hence also \emph{invertible}) gives us the result.
\end{proof}

\subsection{Omitted Proofs from Subsection \ref{subsec:principal-prob}}

\begin{lemma}\label{lem:linear-obj-quadratic-constr}
Let $Q \in \bbR^{d \times d}$ a symmetric $\PD$ matrix and $c$ a vector in $\bbR^d$. Then, the following optimization problem: 
\begin{align*}
    \max_{x \in \bbR^d} \; &c^\top x \\
    \text{s.t.,} \; & x^\top Q x \leq b
\end{align*}
has unique solution: \[x = \frac{b \, Q^ {-1} c}{\sqrt{c^\top Q^{-1} c}}\]
\end{lemma}

\begin{proof}
We first compute the Lagrangian:
\begin{equation}
    L(x, \lambda) = - c^\top x + \frac{\lambda}{2} \left( x^\top Q x - b \right) 
\end{equation}
We can then find the KKT conditions:
\begin{align}
    - c + \lambda Q x &= 0 \label{eq:KKT1}\\
    \lambda &\geq 0 \label{eq:KKT2}\\ 
    \lambda \left( x^\top Q x - b\right) &= 0 \label{eq:KKT3}\\
    x^\top Q x &\leq b \label{eq:KKT4}
\end{align}
At maximum it must be the case that $\lambda > 0$ (from Eq.~\eqref{eq:KKT2}) and hence, combining Eqs.~\eqref{eq:KKT4} and~\eqref{eq:KKT3} we get $x^\top Q x = b$. Due to the fact that $\lambda > 0$, then from Eq.~\eqref{eq:KKT1}, solving in terms of $x$ and using the fact that $Q$ is symmetric positive definite we get: 
\begin{equation}\label{eq:Delta-bef-sub}
    x = \frac{1}{\lambda} Q^{-1} c
\end{equation}
Substituting the above in equation $x^\top Q x = b$ we obtain: 
\begin{equation}
    x^\top Q x = \frac{1}{\lambda^2} c^\top Q^{-1} c = b
\end{equation}
Solving this in terms of $\lambda$ gives $\lambda = \frac{1}{b}\sqrt{c^\top Q^{-1} c}$. Substituting $\lambda$ in Eq.~\eqref{eq:Delta-bef-sub} we get the result. The proof is completed by the fact that the objective function is convex and the feasible set is concave; hence the global optimum is found at a KKT point.
\end{proof}

\begin{proof}[Proof of Lemma~\ref{lem:sw-opt-quasi}]
We first note a useful lemma (which we formally state and prove in Lemma~\ref{lem:linear-obj-quadratic-constr}), namely that if $Q \in \bbR^{d \times d}$ is a symmetric $\PD$ matrix and $c$ a vector in $\bbR^d$ then the solution of the optimization problem $\max_{x} \; c^\top x$ such that $x^\top Q x \leq b$ has unique solution $x = \frac{b \, Q^ {-1} c}{\sqrt{c^\top Q^{-1} c}}$.

Using the closed-form of the agents' best-response from Lemma~\ref{lem:br-computation} in Eq.~\eqref{eq:max-vw} we get that:
\begin{align*}
\vw_{\SW} &= \arg\max_{\vw': \Vert\vw'\Vert_2 \leq 1} \displaystyle\mathop{\E}_{\vx \sim\mathcal{D}_1}\left[ \left \langle {\vwst}, \hvx(\vx;1) \right \rangle \right] + \displaystyle\mathop{\E}_{\vx \sim\mathcal{D}_2}   
     \left[ \left \langle {\vwst}, \hvx(\vw; 2) \right \rangle \right] \\
     &= \arg \max_{\vw': \Vert\vw'\Vert_2 \leq 1}\displaystyle\mathop{\E}_{\vx \sim\mathcal{D}_1}   
     \left[ \left \langle {\vwst}, \vx + \Delta_1(\vw)\right \rangle \right] + \displaystyle\mathop{\E}_{\vx \sim\mathcal{D}_2}   
     \left[ \left \langle {\vwst}, \vx + \Delta_2(\vw)\right \rangle \right]\\
  &= \arg\max_{\vw': \Vert\vw'\Vert_2 \leq 1} \left \langle   \left(A_1^{-1}\Pi_1 + A_2^{-1}\Pi_2 \right) \vw', \vwst \right \rangle \numberthis{\label{eq:vw-quasi}}
\end{align*}
We rewrite the objective function to be optimized above as: $[{\vwst}^\top (A_1^{-1} \Pi_1 + A_2^{-1} \Pi_2 ) ] \vw  = [(A_1^{-1} \Pi_1 + A_2^{-1} \Pi_2 )^\top {\vwst}  ]^\top \vw $ and the constraint for $\vw$ remains: $\vw^\top \vw \leq 1$. This problem is an instance of the problem solved in Lemma~\ref{lem:linear-obj-quadratic-constr} for $c =(A_1^{-1} \Pi_1 + A_2^{-1} \Pi_2 )^\top {\vwst}$, $b = 1$ and $Q$ the identity matrix. Substituting $c, Q, b$ in the solution of Lemma~\ref{lem:linear-obj-quadratic-constr} gives the result.
\end{proof}

\section{Generalizing to Multiple Groups} \label{app:gen-multiple-groups}

Let $\calG$ denote the set of all groups, i.e., $\calG = \{1, 2, \dots, m\}$. As is customary in the literature, we use $\calG_{-j}$ to denote all the groups apart from group $j$, i.e., $\calG_{-j} = \{1, 2, \dots, j-1, j+1, \dots, m\}$. In order to explain how the theorem and proposition statements change when $m > 2$, we first outline how the principal's equilibrium rule changes as a result of the presence of $m > 2$ groups. Due to the fact that the estimated rule for each group $g \in \calG$ is: $\vw_{\proxy}(g) = \Pi_g \vw$, then from extending Lemma~\ref{lem:sw-opt-quasi} we have that the principal's equilibrium rule becomes: 
\begin{equation}\label{eq:vwsw-gen}
    \vw_{\SW} = \frac{\left(\Pi_1 A_1^{-1} + \dots + \Pi_m A_m^{-1} \right)\vwst}{\|\left(\Pi_1 A_1^{-1} + \dots + \Pi_m A_m^{-1} \right)\vwst\|}
\end{equation}

We first analyze the do-no-harm objective for the case that $m > 2$ groups are present in the population. The analogue of Theorem~\ref{thm:do-no-harm} for $m > 2$ groups follows.

\begin{theorem}
In equilibrium, there is no negative externality for group $g$ and any $\vwst$ if and only if for all $g\in \calG$, the matrix $\left( \sum_{i \in \calG} A_i^{-1} \Pi_i \right) \Pi_g A_g^{-1} + A_g^{-1} \Pi_g \left( \sum_{i \in \calG} A_i^{-1} \Pi_i \right)$ is $\PSD$. 
\end{theorem}

This means that we can still guarantee that there is no negative externality for any of the groups in equilibrium in the two cases of interest, namely:
\begin{enumerate}
    \item when the cost matrices are proportional to each other, i.e., $A_i = c_{ij} \cdot A_j$ for all $(i,j) \in \calG^2$ and some scalars $c_{ij} > 0$ (analogue of Proposition~\ref{prop:do-no-harm1} for $m > 2$ groups).
    \item when the subspaces $\calS_1, \dots, \calS_m$ are orthogonal (analogue of Proposition~\ref{prop:do-no-harm2} for $m > 2$ groups).
\end{enumerate}
To derive the aforementioned results, the only change in the proofs of Theorem~\ref{thm:do-no-harm}, and Propositions~\ref{prop:do-no-harm1} and~\ref{prop:do-no-harm2} is that $\vw_{\SW}$ should be substituted with the expression in Equation~\eqref{eq:vwsw-gen}.
 
We proceed to discussing total improvement for $m > 2$ groups. We find it useful to present a slightly generalized version of the \emph{overlap proxy}.

\begin{definition}
Given a scoring rule $\vw \in \bbR^d$ and projections $\Pi_1, \dots, \Pi_m \in \bbR^d$, we define the \emph{overlap proxy} between any two groups $G_i, G_k$ with respect to $\vw$ to be: $r_{i,k}(\vw) \triangleq \| \Pi_i \vw - \Pi_k \vw \|$.
\end{definition}

Using this definition, we can state the direct generalization of Lemma~\ref{lem:overlap-proxy}.

\begin{lemma}
Let $\diff_{j,k} \triangleq |\calI_j(\vw) - \calI_k(\vw)|$ be the disparity in total improvement across groups when the principal's rule is $\vw$. In equilibrium, if $A_j = A_k = \mathbb{I}_{d\times d}$, then: $\diff_{j,k} (\vw_{\SW}) \leq r_{j,k} (\vwst)$. Further, the equality holds if and only if $\Pi_j \vwst$ and $\Pi_k \vwst$ are co-linear.
\end{lemma}

The analogue of Theorem~\ref{thm:eq-imp} for $m > 2$ becomes: 

\begin{theorem}
In equilibrium, the groups obtain equal total improvement for all $\vwst$ if and only if $A_1^{-1} \Pi_1 A_1^{-1} = A_2^{-1} \Pi_2 A_2^{-1} = \dots = A_m^{-1} \Pi_m A_m^{-1}$.
\end{theorem}

Finally, we turn our attention to the per-unit improvement, and we state the analogue of Theorem~\ref{thm:suf-and-nec-per-unit} for $m > 2$ groups. This analogue is again derived using Equation~\eqref{eq:vwsw-gen} for $\vw_{\SW}$.

\begin{theorem}
In equilibrium, group $g$ gets optimal per-unit improvement if and only if: 
\[
\left\langle A_g^{-1}\frac{\Pi_g A_g^{-1}  \vwst}{\left\Vert\Pi_g A_g^{-1} \vwst \right\Vert_2} - A_g^{-1}\frac{\Pi_g\left( \Pi_1 A_1^{-1} +  \Pi_2 A_2^{-1} + \dots + \Pi_m A_m^{-1} \right) \vwst}{\left\Vert\Pi_g\left(\Pi_1 A_1^{-1}  + \Pi_2 A_2^{-1} + \dots + \Pi_m A_m^{-1}\right) \vwst \right\Vert_2},\vwst\right\rangle = 0.
\]
\end{theorem}

Note that this means that, in equilibrium, optimal per-unit outcome improvement is guaranteed if there exists $c_g > 0$, such that: \[ \Pi_g (A_g^{-1} \Pi_g )^\top \vwst = c_g \Pi_g (A_1^{-1} \Pi_1 + \dots + A_m^{-1} \Pi_m )^\top \vwst \] Two notable examples for which this condition holds are:
\begin{enumerate}
    \item when all of $\calS_1, \dots, \calS_m$ are orthogonal to each other
    \item when $A_i = c_{ij} \cdot A_j$ and $\Pi_i = \Pi_j$.
\end{enumerate}

\section{Supplementary Material for Section~\ref{sec:characterizations}}\label{app:characterizations}

\subsection{Omitted Proofs from Subsection \ref{subsec:do-no-harm}} \label{app:do-no-harm}

\begin{proof}[Proof of Theorem~\ref{thm:do-no-harm}]
By Definition \ref{def:per-unit-outcome-improve}, having no negative externality in equilibrium translates to:
\begin{align*}
    \forall g: \; &\calI_g \left(\vw\right) \geq 0 \Leftrightarrow \left\langle A_g^{-1} \Pi_g\frac{\left(A_1^{-1} \Pi_1 + A_2^{-1} \Pi_2 \right)^\top \vwst}{\left\|\left(A_1^{-1} \Pi_1 + A_2^{-1} \Pi_2 \right)^\top \vwst \right\|_2},\vwst\right\rangle \geq 0  &\tag{Lemma~\ref{lem:sw-opt-quasi}} \\
    & \Leftrightarrow \left\langle A_g^{-1} \Pi_g\left(A_1^{-1} \Pi_1 + A_2^{-1} \Pi_2 \right)^\top \vwst,\vwst\right\rangle \geq 0 \\
    &\Leftrightarrow \left\langle \left(\Pi_g^\top {A_g^{-1}}^\top \right)^\top \left(A_1^{-1} \Pi_1 + A_2^{-1} \Pi_2 \right)^\top \vwst,\vwst\right\rangle \geq 0  &\tag{$(A B)^\top = B^\top A^\top$ and ${A^\top}^\top = A$} \\
    &\Leftrightarrow \left\langle \left(\left(A_1^{-1} \Pi_1 + A_2^{-1} \Pi_2\right)\left(\Pi_g^\top {A_g^{-1}}^\top \right) \right)^\top \vwst,\vwst\right\rangle \geq 0  &\tag{$(A B)^\top = B^\top A^\top$}\\
    &\Leftrightarrow \left\langle \left(A_1^{-1} \Pi_1\Pi_g^\top {A_g^{-1}}^\top + A_2^{-1} \Pi_2\Pi_g^\top {A_g^{-1}}^\top \right)^\top \vwst,\vwst\right\rangle \geq 0  &\tag{$(A+B)C = AC + BC$}
\end{align*}
Using the fact that $\Pi_g^\top = \Pi_g$ (as orthogonal projection matrices) and that ${A_g^{-1}}^\top = A_g^{-1}$ (as $A_g$ is a symmetric matrix), we obtain the condition $q(\vwst) \geq 0$ where $q(\vwst) = (\vwst)^\top M \vwst$ is a quadratic form with $M = A_1^{-1} \Pi_1\Pi_g {A_g^{-1}} + A_2^{-1} \Pi_2\Pi_g {A_g^{-1}}$. By standard linear algebra arguments, noting that $(\vwst)^\top M \vwst = ((\vwst)^\top M \vwst)^\top = (\vwst)^\top M^\top \vwst$, we can rewrite $q(\vwst) = \frac{1}{2} (\vwst)^\top \left(M+M^\top\right) \vwst$. The condition then holds for all $\vwst$ if and only if 
\[
M + M^\top 
= A_1^{-1} \Pi_1\Pi_g A_g^{-1} + A_2^{-1} \Pi_2\Pi_g A_g^{-1} + A_{g}^{-1} \Pi_g\Pi_1 A_1^{-1} + A_{g}^{-1} \Pi_g\Pi_2 A_2^{-1}
\]
is $\PSD$. This concludes the proof.
\end{proof}

\begin{proof}[Proof of Corollary~\ref{prop:do-no-harm1}]
Fix a group $g \in \{1,2\}$ (wlog, let $g = 1$), and let $\bA = A_1^{-1}$. Then, from Theorem~\ref{thm:do-no-harm} no negative externality for group $g$ is guaranteed if and only if: 
\begin{align*} 
    \left \langle \left( \bA \, \Pi_1 \, \bA + \frac{1}{c} \bA \, \Pi_2 \Pi_1 \bA \, \right) , \vwst \right \rangle  \geq 0 & \Leftrightarrow  \left( \left( \bA \, \Pi_1 \, \bA + \frac{1}{c} \bA \, \Pi_2 \Pi_1 \, \bA \right)^\top \, \vwst\right)^\top \vwst \geq 0\\
    & \Leftrightarrow  \; {\vwst}^\top \, \left( \bA \, \Pi_1 \, \bA + \frac{1}{c^2} \bA \, \Pi_2 \Pi_1 \, \bA \right) \vwst \geq 0 \numberthis{\label{eq:do-no-harm-simple}}
\end{align*}

Eq.~\eqref{eq:do-no-harm-simple} is true if and only if matrix $\bA \Pi_1 \bA + \bA \Pi_2 \Pi_1 \bA$ is \PSD. Matrix $\Pi_1$ is by definition \PSD. Matrix $A$ is \PD, hence its inverse, $\bA$, is also \PD. As a result, matrix $\bA \Pi_1 \bA$ is \PSD. We shift our attention to matrix $\bA \Pi_2 \Pi_1 \bA$ now. Since $\Pi_1, \Pi_2$ are projection matrices, then the eigenvalues of their product $\Pi_2 \Pi_1$ are non-negative~\citep{ProdProj}. Recall that a matrix is $\PSD$ if and only if its eigenvalues are non-negative. As a result, matrix $\Pi_2 \Pi_1$ is \PSD. Using the same property as above (i.e., that if matrices $A,B$ are \PSD, then so is matrix $A B A$) we can conclude that $\bA \Pi_2 \Pi_1 \bA$ is \PSD. If matrices $A,B$ are \PSD, then so is matrix $A + B$. Hence, matrix $\bA \Pi_2 \bA + \bA \Pi_2 \Pi_2 \bA$ is \PSD, i.e., by definition that for any vector $z$ we have that: $z^\top (\bA \Pi_1 \bA + \bA \Pi_2 \Pi_1 \bA / c) z \geq 0$. This concludes our proof.
\end{proof}

\begin{proof}[Proof of Corollary~\ref{prop:do-no-harm2}]
Fix a group $g \in \{1,2\}$ (wlog let $g = 1$). From Theorem~\ref{thm:do-no-harm} we need: 
\begin{align*}
    \left \langle \left(A_1^{-1} \Pi_1^2 A_1^{-1} + A_2^{-1} \Pi_2 \Pi_1 A_1^{-1} \right)^\top \vwst, \vwst \right \rangle \geq 0 \Leftrightarrow \left \langle \left(A_1^{-1} \Pi_1 A_1^{-1} \right)^\top \vwst, \vwst \right \rangle \geq 0 \numberthis{\label{eq:cor4.2}}
\end{align*}
where for the last inequality we used $\Pi_2 \Pi_1 = 0$ (as subspaces $\calS_1, \calS_2$ are orthogonal) and $\Pi_g^2 = \Pi_g$ (as orthogonal projection matrices). Eq.~\eqref{eq:cor4.2} holds since matrices $\Pi_1$ and $A_1$ are \PSD.
\end{proof}

\subsection{Omitted Proofs from Subsection \ref{subsec:total-outcome-improvement}} \label{app:equal-improvement}

\begin{proof} [Proof of Lemma \ref{lem:overlap-proxy}]
Note that 
\begin{align*}
\calI_1\left(\frac{\left(\Pi_1 + \Pi_2\right) \vwst}{\Vert \left(\Pi_1 + \Pi_2\right) \vwst \Vert}\right) - \calI_2\left(\frac{\left(\Pi_1 + \Pi_2\right) \vwst}{\Vert \left(\Pi_1 + \Pi_2\right) \vwst \Vert}\right) 
&= \left\langle \left(\Pi_1 - \Pi_2 \right) \frac{\left(\Pi_1 + \Pi_2\right) \vwst}{\Vert \left(\Pi_1 + \Pi_2\right) \vwst \Vert}, \vwst \right\rangle
\\&= \frac{1}{\Vert \left(\Pi_1 + \Pi_2\right) \vwst\Vert} \cdot \left\langle \left(\Pi_1 + \Pi_2\right) \vwst, \left(\Pi_1 - \Pi_2 \right) \vwst \right\rangle.
\end{align*}

By Cauchy-Schwarz, we have that 
\begin{align*}
&\left \vert \frac{1}{\Vert \left(\Pi_1 + \Pi_2\right) \vwst\Vert} \cdot \left\langle \left(\Pi_1 + \Pi_2\right) \vwst, \left(\Pi_1 - \Pi_2 \right) \vwst \right\rangle \right \vert
\\&\leq \frac{1}{\Vert \left(\Pi_1 + \Pi_2\right) \vwst\Vert} \cdot 
\Vert \left(\Pi_1 + \Pi_2\right) \vwst \Vert \cdot \Vert \left(\Pi_1 - \Pi_2\right) \vwst \Vert
\\&= \Vert \left(\Pi_1 - \Pi_2\right) \vwst \Vert
\\&= r_{1,2}(\vwst),
\end{align*}
with equality if and only if $\left(\Pi_1 + \Pi_2\right) \vwst$ and $\left(\Pi_1 - \Pi_2 \right) \vwst$ are colinear, i.e., there exists $\alpha \in \mathbb{R}$ such that $\alpha \left(\Pi_1 + \Pi_2\right) \vwst = \left(\Pi_1 - \Pi_2 \right) \vwst$, which can be equivalently written as $(1-\alpha) \Pi_1 \vwst = (1+\alpha) \Pi_2 \vwst$, i.e., $\Pi_1 \vwst$ and $\Pi_2 \vwst$ are colinear.
\end{proof}

\begin{proof}[Proof of Theorem~\ref{thm:eq-imp}]
Equal total outcome improvement across groups is guaranteed in equilibrium if and only if the following holds:
\begin{align*}
&\calI_1 \left( \vw_{\SW} \right) - \calI_2 \left( \vw_{\SW}\right) = 0 \Leftrightarrow &\tag{Definition~\ref{def:per-unit-outcome-improve}}\\
\Leftrightarrow &\left\langle A_1^{-1}\Pi_1\frac{ \left(A_1^{-1} \Pi_1 + A_2^{-1} \Pi_2 \right)^\top \vwst}{\left\|\left(A_1^{-1} \Pi_1 + A_2^{-1} \Pi_2 \right)^\top \vwst \right\|_2} - A_2^{-1}\Pi_2\frac{ \left(A_1^{-1} \Pi_1 + A_2^{-1} \Pi_2 \right)^\top \vwst}{\left\|\left(A_1^{-1} \Pi_1 + A_2^{-1} \Pi_2 \right)^\top \vwst \right\|_2},\vwst\right\rangle = 0  \\
\Leftrightarrow&\left\langle A_1^{-1}\Pi_1 \left(A_1^{-1} \Pi_1 + A_2^{-1} \Pi_2 \right)^\top \vwst - A_2^{-1}\Pi_2 \left(A_1^{-1} \Pi_1 + A_2^{-1} \Pi_2 \right)^\top \vwst,\vwst\right\rangle = 0 \\
\Leftrightarrow &\left\langle \left[A_1^{-1}\Pi_1 \left(A_1^{-1} \Pi_1 + A_2^{-1} \Pi_2 \right)^\top - A_2^{-1}\Pi_2 \left(A_1^{-1} \Pi_1 + A_2^{-1} \Pi_2 \right)^\top\right] \vwst,\vwst\right\rangle = 0 \\
\Leftrightarrow &\left\langle\left(A_1^{-1}\Pi_1\Pi_1^\top {A_1^{-1}}^\top + A_2^{-1}\Pi_2\Pi_1^\top {A_1^{-1}}^\top - A_1^{-1}\Pi_1\Pi_2^\top {A_2^{-1}}^\top - A_2^{-1}\Pi_2\Pi_2^\top {A_2^{-1}}^\top\right)^\top \vwst,\vwst\right\rangle = 0 \\
\Leftrightarrow &\left\langle\left(A_1^{-1}\Pi_1 {A_1^{-1}}+ A_2^{-1}\Pi_2\Pi_1 {A_1^{-1}}^\top - A_1^{-1}\Pi_1\Pi_2 {A_2^{-1}} - A_2^{-1}\Pi_2 {A_2^{-1}}\right)^\top \vwst,\vwst\right\rangle = 0
\end{align*}
where the second transition is due to Lemma~\ref{lem:sw-opt-quasi}, the fourth is due to $Av - Bv = (A-B)v$, the second-to-last one is due to $(A B)^\top = B^\top A^\top$ and ${A^\top}^\top = A$, and the last one is due to the fact that $\Pi_g = \Pi_g^\top$, $\Pi_g^2 = \Pi_g$ and $A_g^{-1} = {A_g^{-1}}^\top$. 
Let $M \triangleq A_1^{-1}\Pi_1 {A_1^{-1}}+ A_2^{-1}\Pi_2\Pi_1 {A_1^{-1}} - A_1^{-1}\Pi_1\Pi_2 {A_2^{-1}} - A_2^{-1}\Pi_2 {A_2^{-1}}$, the above can be written as the quadratic form $q(\vwst) = (\vwst)^\top M \vwst$. In turn, $q(\vwst) = 0$ \emph{simultaneously for all $\vwst$}, i.e $q = 0$, if and only $M$ is skew-symmetric, which means $M + M^\top = 0$. This can be rewritten as 
\begin{align*}
0 &= A_1^{-1}\Pi_1 {A_1^{-1}}+ A_2^{-1}\Pi_2\Pi_1 {A_1^{-1}}^\top - A_1^{-1}\Pi_1\Pi_2 {A_2^{-1}} - A_2^{-1}\Pi_2 {A_2^{-1}} 
\\& + A_1^{-1}\Pi_1 {A_1^{-1}} + A_1^{-1}\Pi_1\Pi_2 {A_1^{-1}}^\top - A_2^{-1}\Pi_2\Pi_1 {A_1^{-1}} - A_2^{-1}\Pi_2 {A_2^{-1}} ,
\end{align*}
or equivalently: $2 A_1^{-1}\Pi_1 {A_1^{-1}} - 2 A_2^{-1}\Pi_2 {A_2^{-1}} = 0$.
This concludes the proof.
\end{proof}

\subsection{Omitted Proofs from Section \ref{subsec:per-unit-outcome-improvement}} \label{app:opt-per-unit}

\begin{proof}[Proof of Proposition~\ref{prop:rel-improvement-arbitrary}]
We focus on a two-dimensional example for clarity of exposition. To abstract away from discrepancies in the cost matrices, we assume that $A_1, A_2 = \mathbbm{I}_{2 \times 2}$, 
and that the projection matrices of the two groups are 
\begin{equation*}
\Pi_1 = 
    \begin{bmatrix}
    1 & 0 \\
    0 & 0
    \end{bmatrix}
    \quad \text{and} \quad 
    \Pi_2 = 
    \begin{bmatrix}
    0 & 0 \\
    0 & 1
    \end{bmatrix}.
\end{equation*}
Next, we select $\eps > 0$ such that $\frac{\eps^2}{1-\eps^2} < \alpha$, and assume that the distribution we face is such that the optimal outcome-decision rule is: $\vwst = \left[\eps, \sqrt{1-\eps^2}\right]^{\top}$. From Lemma~\ref{lem:sw-opt-quasi}, substituting the values of $A_1, A_2, \Pi_1, \Pi_2$ as defined above, we have that the $\vw$ maximizing the social welfare satisfies: 
\[
\vw_{\SW} = \frac{\vwst}{\|\vwst\|_2} = \left[\eps, \sqrt{1-\eps^2}\right]^{\top}.
\]
Substituting $A_1, A_2, \Pi_1, \Pi_2$ in $\Delta_g(\vw) = A_g^{-1}\Pi_g \vw$ we have: $\calI_1(\vw_{\SW}) = \eps^2$ and $\calI_2(\vw_{\SW}) = 1-\eps^2$.

Next, we compute $\ucalI_1(\vw_1)$ and $\ucalI_2(\vw2)$. Substituting $A_1, A_2, \Pi_1, \Pi_2$ in the definition of $\Delta_g(\vw)$ and by Lemma~\ref{lem:group-opt-vw}, we have: $\vw_1 = [1,0]^\top, \vw_2 = [0,1]^\top$. Finally:
\begin{align*}
\ucalI_1(\vw) = \ucalI_1\left(\frac{\Pi_1 \left[\eps, \sqrt{1-\eps^2}\right]^{\top}}{\left\Vert\Pi_1 \left[\eps, \sqrt{1-\eps^2}\right]^{\top}\right\Vert_2}\right) = \ucalI_1\left(\frac{ \left[\eps, 0\right]^{\top}}{\left\Vert \left[\eps, 0\right]^{\top}\right\Vert_2}\right) &= \ucalI_1\left( \left[1, 0\right]^{\top}\right) = \eps = \max_{\vw'}\ucalI_1(\vw')
\end{align*}
\begin{align*}
\ucalI_2(\vw) = \ucalI_1\left(\frac{\Pi_2 \left[\eps, \sqrt{1-\eps^2}\right]^{\top}}{\left\Vert\Pi_2 \left[\eps, \sqrt{1-\eps^2}\right]^{\top}\right\Vert_2}\right) = \ucalI_2\left(\frac{ \left[0, \sqrt{1-\eps^2}\right]^{\top}}{\left\Vert \left[0, \sqrt{1-\eps^2}\right]^{\top}\right\Vert_2}\right) &= \ucalI_2\left( \left[0, 1\right]^{\top}\right) = \sqrt{1-\eps^2} \\
&= \ucalI_2(\vw_2).    
\end{align*}

However, for the total outcome improvement: $\frac{\calI_1(\vw)}{\calI_2(\vw)} = \frac{\eps^2}{1-\eps^2} < \alpha$, which concludes the proof.
\end{proof}

\begin{proof}[Proof of Theorem~\ref{thm:suf-and-nec-per-unit}]
Using Lemmas~\ref{lem:sw-opt-quasi},~\ref{lem:group-opt-vw} and Definition~\ref{def:per-unit-outcome-improve} we get that $\vw_{\SW}$ induces optimal per-unit outcome improvement if and only if:
\begin{align*}
    &\vw_{\SW} = \vw_g = \arg\max_{\vw'} \, \ucalI_g \left( \vw' \right) \Leftrightarrow \\
    \Leftrightarrow \; & \ucalI_g \left( \vw_g\right) - \ucalI_g \left( \vw\right) = 0 \Leftrightarrow \\
    \Leftrightarrow \; & \calI_g \left(\frac{\Pi_g\vw_g}{\left\Vert\Pi_g\vw_g\right\Vert_2}\right) - \calI_g\left(\frac{\Pi_g\vw}{\left\Vert\Pi_g\vw\right\Vert_2}\right) = 0 &\tag{Definition of $\ucalI_g(\cdot)$}\\
    \Leftrightarrow \; &\left\langle A_g^{-1} \Pi_g\frac{\frac{\Pi_g\left(A_g^{-1} \Pi_g \right)^\top \vwst}{\left\Vert\left(A_g^{-1} \Pi_g \right)^\top \vwst \right\Vert_2}}{\left\Vert\frac{\Pi_g\left(A_g^{-1} \Pi_g \right)^\top \vwst}{\left\Vert\left(A_g^{-1} \Pi_g \right)^\top \vwst \right\Vert_2}\right\Vert_2} - A_g^{-1} \Pi_g\frac{\frac{\Pi_g\left(A_1^{-1} \Pi_1 + A_2^{-1} \Pi_2 \right)^\top \vwst}{\left\Vert\left(A_1^{-1} \Pi_1 + A_2^{-1} \Pi_2 \right)^\top \vwst \right\Vert_2}}{\left\Vert\frac{\Pi_g\left(A_1^{-1} \Pi_1 + A_2^{-1} \Pi_2 \right)^\top \vwst}{\left\Vert\left(A_1^{-1} \Pi_1 + A_2^{-1} \Pi_2 \right)^\top \vwst \right\Vert_2}\right\Vert_2},\vwst\right\rangle = 0 &\tag{4}\\
    \Leftrightarrow \; &\left\langle A_g^{-1} \Pi_g\frac{\frac{\Pi_g\left(A_g^{-1} \Pi_g \right)^\top \vwst}{\left\Vert\left(A_g^{-1} \Pi_g \right)^\top \vwst \right\Vert_2}}{\frac{\Vert\Pi_g\left(A_g^{-1} \Pi_g \right)^\top \vwst\Vert_2}{\left\Vert\left(A_g^{-1} \Pi_g \right)^\top \vwst \right\Vert_2}} - A_g^{-1} \Pi_g\frac{\frac{\Pi_g\left(A_1^{-1} \Pi_1 + A_2^{-1} \Pi_2 \right)^\top \vwst}{\left\Vert\left(A_1^{-1} \Pi_1 + A_2^{-1} \Pi_2 \right)^\top \vwst \right\Vert_2}}{\frac{\Vert\Pi_g\left(A_1^{-1} \Pi_1 + A_2^{-1} \Pi_2 \right)^\top \vwst\Vert_2}{\left\Vert\left(A_1^{-1} \Pi_1 + A_2^{-1} \Pi_2 \right)^\top \vwst \right\Vert_2}},\vwst\right\rangle = 0 &\tag{5}\\
    \Leftrightarrow \; &\left\langle A_g^{-1}\Pi_g\frac{\Pi_g\left(A_g^{-1} \Pi_g \right)^\top \vwst}{\left\Vert\Pi_g\left(A_g^{-1} \Pi_g \right)^\top \vwst \right\Vert_2} - A_g^{-1}\Pi_g \frac{\Pi_g\left(A_1^{-1} \Pi_1 + A_2^{-1} \Pi_2 \right)^\top \vwst}{\left\Vert\Pi_g\left(A_1^{-1} \Pi_1 + A_2^{-1} \Pi_2 \right)^\top \vwst \right\Vert_2},\vwst\right\rangle = 0 \\
    \Leftrightarrow \;& \left\langle A_g^{-1}\frac{\Pi_g\left(A_g^{-1} \Pi_g \right)^\top \vwst}{\left\Vert\Pi_g\left(A_g^{-1} \Pi_g \right)^\top \vwst \right\Vert_2} - A_g^{-1} \frac{\Pi_g\left(A_1^{-1} \Pi_1 + A_2^{-1} \Pi_2 \right)^\top \vwst}{\left\Vert\Pi_g\left(A_1^{-1} \Pi_1 + A_2^{-1} \Pi_2 \right)^\top \vwst \right\Vert_2},\vwst\right\rangle = 0 
    &\tag{7}
\end{align*}
where the transitions are given by: (4)\; Lemmas~\ref{lem:sw-opt-quasi} and~\ref{lem:group-opt-vw}, \; (5)\; $\Vert\frac{\mathbf{v}}{c}\Vert_2 = \frac{\Vert\mathbf{v}\Vert_2}{c}$ for any scalar $c$, and (7)\; $\Pi_g\Pi_g = \Pi_g$ as they are orthogonal projections.
\end{proof}

\begin{lemma}\label{lem:suf-per-unit}
In equilibrium, optimal per-unit outcome improvement is guaranteed if there exists $c_g > 0$, such that:
\[
\Pi_g(A_g^{-1}\Pi_g)^\top\vwst = c_g\Pi_g(A_1^{-1}\Pi_1+A_2^{-1}\Pi_2)^\top\vwst.
\]
\end{lemma}

\begin{proof}[Proof of Lemma \ref{lem:suf-per-unit}]
Assume the condition in the statement holds and denote
\[
v = \Pi_g(A_g^{-1}\Pi_g)^\top\vwst = c_g\Pi_g(A_1^{-1}\Pi_1+A_2^{-1}\Pi_2)^\top\vwst.
\]

By Theorem \ref{thm:suf-and-nec-per-unit}, we know that for any group $g\in\{1,2\}$, we are guaranteed Optimal per-unit outcome Improvement if and only if the following holds:
\[
\left\langle A_g^{-1}\frac{\Pi_g\left(A_g^{-1} \Pi_g \right)^\top \vwst}{\left\Vert\Pi_g\left(A_g^{-1} \Pi_g \right)^\top \vwst \right\Vert_2} - A_g^{-1} \frac{\Pi_g\left(A_1^{-1} \Pi_1 + A_2^{-1} \Pi_2 \right)^\top \vwst}{\left\Vert\Pi_g\left(A_1^{-1} \Pi_1 + A_2^{-1} \Pi_2 \right)^\top \vwst \right\Vert_2},\vwst\right\rangle = 0.
\]
Which is in our case equivalent to requiring 
\[
\left\langle A_g^{-1}\frac{v}{\left\Vert v\right\Vert_2} - A_g^{-1} \frac{\frac{v}{c_g}}{\left\Vert\frac{v}{c_g} \right\Vert_2},\vwst\right\rangle = 0.
\]
Equivalently, this can be written as
\[
\left\langle A_g^{-1}\frac{v}{\left\Vert v\right\Vert_2} - A_g^{-1} \frac{v}{\left\Vert v \right\Vert_2},\vwst\right\rangle = 0.
\]
This concludes the proof.
\end{proof}

\begin{proof}[Proof of Corollary \ref{cor:optimal-per-unit-1}]
The condition in Lemma \ref{lem:suf-per-unit} can be written as
\[
\Pi_gA_g^{-1}\vwst = c_g\left(\Pi_g\Pi_1A_1^{-1}+\Pi_g\Pi_1A_2^{-1}\right)\vwst
\]

Since $\Pi_1\Pi_2 = \Pi_2\Pi_1 = 0$, the condition holds for 
\[
c_1 = c_2 = 1
\]
\end{proof}

\begin{proof}[Proof of Corollary \ref{cor:optimal-per-unit-2}]
Assume $\Pi_1=\Pi_2$ and $A_1 = cA_2$. It is immediate that the condition of Lemma \ref{lem:suf-per-unit} holds for 
\[
c_1 = \frac{1}{1+c},\quad c_2 = \frac{c}{1+c}
\]
\end{proof}

\subsection{Omitted Proofs from Section \ref{subsec:joint}} \label{app:joint}

As in many practical settings, the principal is expected to take into account a joint objective of predictive accuracy and social welfare, We show next that under mild conditions, deploying any combination of the social-welfare maximizing solution and the true underlying predictor results in inheriting the Do-No-Harm guarantee of the social welfare maximizer.

We begin by proving three useful lemmas:
\begin{lemma}\label{lem:linearcomb}
Assume Do-No-Harm is guaranteed for group $g \in G$ for each of $\vw^1$,$\vw^2$. Then, Do-No-Harm is guaranteed for $\alpha \vw^1+\beta \vw^2$, for any $\alpha,\beta \geq 0$.
\end{lemma}

\begin{proof} [Proof of Lemma \ref{lem:linearcomb}] 
\begin{align*}
\calI_g\left(\alpha\vw^1 + \beta\vw^2\right) &= \left\langle \Delta_g\left( \alpha\vw^1 + \beta\vw^2\right),\vwst \right\rangle &\tag{Definition \ref{def:per-unit-outcome-improve}}\\
&= \left\langle A_g^{-1}\Pi_g\left(\alpha\vw^1 + \beta\vw^2\right),\vwst \right\rangle &\tag{Lemma \ref{lem:br-computation}}\\
&= \alpha\left\langle A_g^{-1}\Pi_g\vw^1,\vwst \right\rangle + \beta\left\langle A_g^{-1}\Pi_g\vw^2,\vwst \right\rangle &\tag{Linearity} \\
&= \alpha \calI_g\left(\vw^1\right) + \beta \calI_g\left(\vw^2\right) &\tag{Definition \ref{def:per-unit-outcome-improve}}\\
&\geq 0 &\tag{$\calI_g\left(\vw^1\right),\calI_g\left(\vw^2\right)\geq0$}
\end{align*}
\end{proof}

\begin{lemma}\label{lem:do-no-harm-acc}
Assume for some $g \in G$, $\Pi_gA_g^{-1}$ is positive semi-definite. Then Do-No-Harm is guaranteed for group $g$ when the principal deploys $\vwst$.
\end{lemma}

\begin{proof} [Proof of Lemma \ref{lem:do-no-harm-acc}] 
\begin{align*}
\calI_g\left(\vwst\right) &= \left\langle \Delta_g\left( \vwst\right),\vwst \right\rangle &\tag{Definition \ref{def:per-unit-outcome-improve}}\\
&= \left\langle A_g^{-1}\Pi_g\vwst,\vwst \right\rangle &\tag{Lemma \ref{lem:br-computation}}\\
&= {\vwst}^\top\Pi_gA_g^{-1}\vwst &\tag{By definition} \\
&\geq 0 &\tag{$\Pi_gA_g^{-1}\succcurlyeq 0$}
\end{align*}
\end{proof}

\begin{lemma}\label{lem:commute}
Assume for some $g \in G$, $\Pi_g$,$A_g^{-1}$ commute. Then $\Pi_gA_g^{-1}$ is positive semi-definite.
\end{lemma}

\begin{proof} [Proof of Lemma \ref{lem:commute}] 
For any $\vw$,
\begin{align*}
\vw^\top\Pi_gA_g^{-1}\vw 
&= \vw^\top\Pi_g\Pi_gA_g^{-1}\vw &\tag{$\Pi_g = \Pi_g\Pi_g$} \\
&= \vw^\top\Pi_g^\top A_g^{-1}\Pi_g\vw &\tag{$\Pi_g = \Pi_g^\top, \Pi_gA_g^{-1} = A_g^{-1}\Pi_g$}\\
&= \left(\Pi_g\vw\right)^\top A_g^{-1}\Pi_g\vw &\tag{$\vw^\top\Pi_g^\top = \left(\Pi_g\vw\right)^\top$}\\
&\geq 0 &\tag{$A_g^{-1} \succcurlyeq 0$}
\end{align*}
\end{proof}

\begin{observation}
Note that the assumption that $\Pi_gA_g^{-1}$ is positive semi-definite is very intuitive. The reason is that the information recovered by group $g$ regarding the principal's decision rule $\vw$ resides in the subspace defined by $\Pi_g$. 
\end{observation}

\begin{proof} [Proof of Theorem \ref{thm:joint}] 
The theorem follows directly by combining Lemmas \ref{lem:decomposible},
\ref{lem:linearcomb}, \ref{lem:do-no-harm-acc}, and \ref{lem:commute}.
\end{proof}

\begin{proof} [Proof of Corollary \ref{cor:full}]
Given the fact that for all $g$, $\Pi_g$ and $A_g^{-1}$ commute since $\Pi_g = I$, the corollary follows directly from Lemmas \ref{lem:linearcomb}, \ref{lem:do-no-harm-acc}, and \ref{lem:commute}.
\end{proof}

\section{On \texorpdfstring{$\Pi_g$ And $A_g^{-1}$} Commuting} \label{app:commuting}

\begin{lemma} \label{lem:decomposible}
Suppose that for all $x \in S_g$, $A_g x \in S_g$, and for all $x \in S_g^\bot$, $A_g x \in S_g^\bot$. Then, $A_g^{-1}$ and $\Pi_g$ commute. Further,
\[
\COST\left(\vx, \vx';g \right) = (\Pi_g (\vx'-\vx))^\top A_g (\Pi_g (\vx'-\vx)) + (\Pi_g^\bot (\vx'-\vx))^\top A_g (\Pi_g^\bot (\vx'-\vx)).
\]
\end{lemma}

\begin{proof}
Remember that because $S_g^\bot$ is the orthogonal subspace to $S_g$, we have that for all $x$, $x = \Pi_g x + \Pi_g^\bot x$. Now, note that
\[
\Pi_g A_g x = \Pi_g A_g (\Pi_g x) + \Pi_g A_g (\Pi_g^\bot x).
\]
Since $\Pi_g x \in S_g$, we have that $A_g (\Pi_g x) \in S_g$, hence $\Pi_g A_g (\Pi_g x) = A_g (\Pi_g x)$ (since $\Pi_g$ is the orthogonal projection operator onto $S_g$). Further, since $\Pi_g^\bot x \in S_g^\bot$, we have that $A_g (\Pi_g^\bot x) \in S_g^\bot$, leading to $\Pi_g A_g (\Pi_g^\bot x) = 0$. This leads to $\Pi_g A_g x = A_g \Pi_g x$ for all $x$, directly implying that $\Pi_g A_g = A_g \Pi_g$. This can be further rewritten as 
\[
\Pi_g = A_g \Pi_g A_g^{-1}, 
\]
or equivalently 
\[
A_g^{-1} \Pi_g = \Pi_g A_g^{-1},
\]
showing the first part of the result.

The second part of the result follows immediately from
\begin{align*}
\COST\left(\vx, \vx';g \right) 
&= (\vx'-\vx)^\top A_g (\vx'-\vx)
\\& = \left(\Pi_g (\vx'-\vx) + \Pi_g^\bot (\vx'-\vx)\right)^\top A_g \left(\Pi_g (\vx'-\vx) + \Pi_g^\bot (\vx'-\vx)\right)
\\&= (\Pi_g (\vx'-\vx))^\top A_g (\Pi_g (\vx'-\vx)) + (\Pi_g^\bot (\vx'-\vx))^\top A_g (\Pi_g^\bot (\vx'-\vx)) 
\\&+ (\Pi_g (\vx'-\vx))^\top A_g (\Pi_g^\bot (\vx'-\vx)) + (\Pi_g^\bot (\vx'-\vx))^\top A_g (\Pi_g (\vx'-\vx)).
\end{align*}
Since $A_g (\Pi_g^\bot (\vx'-\vx)) \in S_g^\bot$, we have that $\Pi_g (\vx'-\vx))^\top A_g (\Pi_g^\bot (\vx'-\vx)) = (\vx'-\vx)^\top \Pi_g A_g (\Pi_g^\bot (\vx'-\vx)) = 0$. A similar argument holds to show $(\Pi_g^\bot (\vx'-\vx))^\top A_g (\Pi_g (\vx'-\vx)) = 0$, concluding the proof.
\end{proof}

Note that this condition is a fairly natural one: indeed, it shows that the cost of modifying a feature vector $x$ can be decomposed in two independent components: modifying $\Pi_g x$, the part of $x$ that is in $S_g$, and modifying $\Pi_g^\bot x$, the part of $x$ that is in $S_g^\bot$. In turn, it indicates that feature modifications within $S_g$ do not affect feature modifications that within $S_g^\bot$, and vice-versa. One reason for this is that agents in group $g$ \emph{are only aware of $\Pi_g x$ in the first place}, and \emph{only consider feature modifications that are entirely contained within $S_g$}; in this case, feature modifications within $S_g^\bot$ do not matter, as they are never considered by the agents in the first place who are only interested in improving $(\Pi_g w)^\top x = w^\top (\Pi_g x)$ (where $w$ is the deployed rule).

\section{Supplementary Material for Section~\ref{sec:experiments}}\label{app:experiments}

In this section, we study the impact of disparities in access to information about the model not just on their own, but in conjunction with cost disparities and asymmetries of the scoring rule. To do so, we provide additional experimental results on both the $\textsc{Adult}$ and the $\textsc{Taiwan-Credit}$ dataset. 

In Figure~\ref{fig:experiment-random}, we study disparities in improvement on the $\textsc{Adult}$ and the $\textsc{Taiwan-Credit}$ datasets. While we keep using the same data $X_g$, rule $\vwst$, and projection matrices $\Pi_g$ as in Section~\ref{sec:experiments}, we now consider non-identity cost matrices. To do so, we draw $A_g$ uniformly at random; for $A_g$'s the uniform distribution is taken over $[-1, 1]$, coefficient by coefficient. 

\begin{figure}[ht]
\centering
\resizebox{\linewidth}{!}{%
\begin{tikzpicture}
        
        \begin{axis}[
        name=ax1,
        ybar,
        enlarge x limits = 0.2,
        bar width=7pt,
        ylabel={Improvement},
        legend style={font=\small, at={(-0.1,0.-0.2)},anchor=north west,legend columns=2},
        symbolic x coords={Age, Education, Marriage},
        ymax=1.7,
        ymin=-0.5,
        xtick=data,
        ytick={-0.5, 0, 0.5, 1,1.5},
        ]

\addlegendentry{$\calI_1(\vw)$}
\addplot[black,fill=blue!40!white,postaction={pattern=north east lines}] plot coordinates {
        (Age,       1.05)
        (Education,   1.51)
        (Marriage,  1.18)
        
};

\addlegendentry{$\calI_2(\vw)$}
\addplot [black, fill=red!60!white, postaction={pattern=grid}]
plot coordinates {
        (Age,       -0.05)
        (Education, 0.27)
        (Marriage,  -0.01)
};

 \addlegendentry{$\ucalI_1(\vw)$}
 \addplot[black, fill=yellow, postaction={pattern=dots}] plot coordinates {
         (Age,       1.07)
         (Education, 1.57)
         (Marriage,  1.21)
 };
 
 \addlegendentry{$\ucalI_2(\vw)$}
 \addplot[black, fill=gray] plot coordinates {
         (Age,       -0.08)
         (Education, 0.46)
         (Marriage,  -0.01)
 };

 \addlegendentry{$\ucalI_1({\vw_1})$}
 \addplot[black, fill=magenta!40!white, postaction={pattern=north west lines}] plot coordinates {
         (Age,       1.11)
         (Education, 1.6)
         (Marriage,  1.25)
 };

 \addlegendentry{$\ucalI_2^(\vw_2)$}
\addplot [black, fill=green, postaction={pattern=horizontal lines}]
 plot coordinates {
         (Age,       0.38)
         (Education, 0.59)
         (Marriage,  0.42)
 };
 
\legend{}

\end{axis}

    \begin{axis}[
        at={(ax1.south east)},
        xshift=2cm,
        ybar,
        enlarge x limits = 0.2,
        bar width=7pt,
        ylabel={Improvement},
        legend style={font=\normalsize, nodes={scale=1, transform shape}, at={(0.6,1.2)},anchor=north east,legend columns=6},         symbolic x coords={Age, Country, Education},
        ymax=1.5,
        ymin=0,
        xtick=data,
        ytick={0, 0.5, 1, 1.6},
        ]

\addlegendentry{$\calI_1(\vw_{\SW})$}
\addplot[black,fill=blue!40!white,postaction={pattern=north east lines}] plot coordinates {
        (Age,           0.35)
        (Country,       0.33)
        (Education,     0.23)
};

\addlegendentry{$\calI_2(\vw_{\SW})$}
\addplot [black, fill=red!60!white, postaction={pattern=grid}]
plot coordinates {
        (Age,       0.61)
        (Country,   0.65)
        (Education, 1.43)
};

 \addlegendentry{$\ucalI_1(\vw_{\SW})$}
 \addplot[black, fill=yellow, postaction={pattern=dots}] plot coordinates {
         (Age,       0.41)
         (Country,   0.36)
         (Education, 0.3)
 };

 \addlegendentry{$\ucalI_2(\vw_{\SW})$}
 \addplot[black, fill=gray] plot coordinates {
         (Age,       0.64)
         (Country,   0.67)
         (Education, 1.46)
 };

 \addlegendentry{$\ucalI_1(\vw_1)$}
 \addplot[black, fill=magenta!40!white, postaction={pattern=north west lines}] plot coordinates {
         (Age,       0.62)
         (Country,   0.57)
         (Education, 0.58)
 };
 
\addlegendentry{$\ucalI_2(\vw_2)$}
\addplot [black, fill=green, postaction={pattern=horizontal lines}]
 plot coordinates {
         (Age,       0.8)
         (Country,   0.81)
         (Education, 1.53)
 };

\end{axis}
\end{tikzpicture} 
}
\caption{Left, Right: evaluation on the TAIWAN-CREDIT and ADULT dataset respectively. $A_g$'s are drawn at random.}\label{fig:experiment-random}
\end{figure}
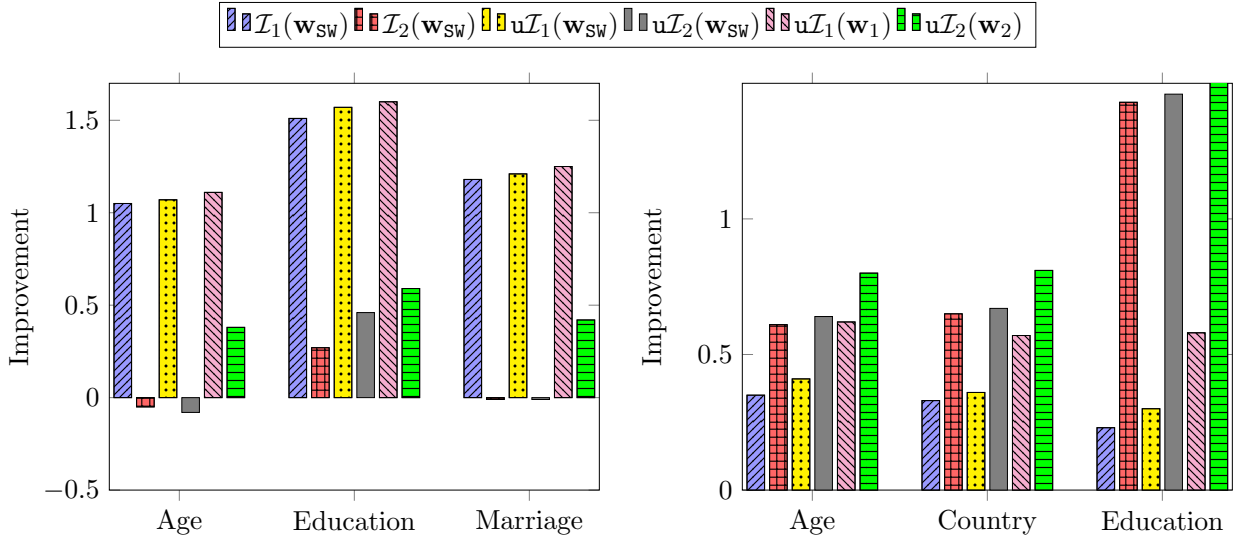

We first note that the scale of the improvements may differ from those of Figure~\ref{fig:experiments}; for example, the difference is striking when looking for example at the ``education'' feature of the left plot, for the TAIWAN-CREDIT dataset. Compared to Figure \ref{fig:experiments} where both total and per-unit outcome improvements are a bit less than $0.5$ for all groups, we see that they are now above $1.5$ for group 2. Significant changes in outcome improvements can also be seen for group 1 on the ``age'' and ``marriage'' feature of the left plot (TAIWAN-CREDIT), and for group 2 on the ``country'' and ``education'' features of the right plot (ADULT). This comes from the fact that information disparities are not the only parameter that have an effect on disparities of improvements across groups: changing the value and magnitude of the $A_g$'s changes which features can be easily modified by agents, and what features give them the best improvement \emph{per level of cost exerted}. This changes which features are desirable to invest in for the agents, and hence for a welfare maximizing principal. 

In the case of the ``Marriage'' feature, because both the total and per-unit improvements are significantly reduced compared to Figure \ref{fig:experiments}, it seems that the disparities we observe are not only due to the fact that the learner may be putting less ``weight'' on group $1$ (defined as the part of the norm of $w_{SW}$ that belongs to group $1$'s information space $S_1$) and more on group $2$; rather, what seems to happen is that the learner focuses on directions in which both groups have information, but that are only ``good'' and easy to modify for group $2$. 

We further observe that the addition of non-identity cost matrices can lead to a \emph{degradation} of outcomes in one of the groups, when the principal optimizes over the joint social welfare. This is visible on the left plot in Figure~\ref{fig:experiment-random}, where the total and per-unit improvements for group $1$ are negative for the ``Age'' feature. This matches the relatively counter-intuitive observation of Section~\ref{subsec:do-no-harm} that optimizing for the social welfare of both groups may hurt the welfare of one of them.

Finally, when comparing the results across age groups in the left plots for Figure \ref{fig:experiments} and Figure \ref{fig:experiment-random} for the ``Age'' feature, we observe a significant reversal of the disparities of improvements across groups: group $1$ was obtaining slightly better outcomes than group $2$ in Figure \ref{fig:experiments}, but group $2$ has significantly worse (in fact, negative) improvements while group $1$ improves slightly more than before in Figure \ref{fig:experiment-random}. This paints a \emph{nuanced} picture that shows that the amount of information that a group has about the scoring rule used by the principal is not the only factor of importance. While having more information is important, how this information interacts with the true model $\vwst$ and the strategic behavior of the agents matters; having a lot of information in directions that have little effect on an agents' true label, or in directions that are very costly for some agents to modify, does not help them when it comes to improving their true labels.

\end{document}